\documentclass[10pt]{article}
\usepackage{amsfonts,amsmath,amssymb,amsthm}
\usepackage[hmargin=1.2in,vmargin=1.2in]{geometry}
\usepackage{natbib} 
\usepackage{setspace} 
\usepackage{enumerate}
\usepackage[shortlabels]{enumitem}
\usepackage{url}
\usepackage{array}
\usepackage[toc,title,titletoc,header]{appendix}
\usepackage[colorlinks,citecolor=magenta]{hyperref}
\usepackage{float}
\usepackage{threeparttable}
\usepackage{adjustbox}
\usepackage{booktabs}
\usepackage{etoolbox} 
\usepackage{mathtools}

\renewcommand{\qed}{\rule{2mm}{2mm}}
\newcommand{\indep}{\perp \!\!\! \perp}
\DeclareMathOperator{\var}{Var}
\DeclareMathOperator{\cov}{Cov}

\allowdisplaybreaks

\newtheorem{theorem}{Theorem}[section]
\newtheorem{lemma}{Lemma}[section]
\newtheorem{corollary}{Corollary}[section]
\theoremstyle{definition}

\newtheorem{remark}{Remark}[section]
\newtheorem{assumption}{Assumption}[section]

\AtEndEnvironment{remark}{~\qed}
\AtEndEnvironment{example}{~\qed}

\begin{document}
\author{
Yuehao Bai\\ 
Department of Economics\\
University of Southern California\\
\href{mailto:yuehao.bai@usc.edu}{\texttt{yuehao.bai@usc.edu}}
\and
Meng Hsuan Hsieh\\
Ross School of Business\\
University of Michigan\\
\href{mailto:rexhsieh@umich.edu}{\texttt{rexhsieh@umich.edu}}
\and
Jizhou Liu\\
Booth School of Business\\
University of Chicago\\
\href{mailto:jliu32@chicagobooth.edu}{\texttt{jliu32@chicagobooth.edu}}
\and
Max Tabord-Meehan\\
Department of Economics\\
University of Chicago\\
\href{mailto:maxtm@uchicago.edu}{\texttt{maxtm@uchicago.edu}}
}


\title{Revisiting the Analysis of Matched-Pair and Stratified Experiments in the Presence of Attrition \thanks{We thank Rachel Glennerster, Hongchang Guo, David McKenzie, Azeem Shaikh, Alex Torgovitsky, Ed Vytlacil, and three anonymous refereees for helpful comments. We also thank  Lorenzo Casaburi and Tristan Reed for helpful comments and for sharing their data for one of our empirical applications. The fourth author acknowledges support from NSF grant SES-2149408.}}

\maketitle

\vspace{-0.3in}

\begin{spacing}{1.2}
\begin{abstract}
In this paper we revisit some common recommendations regarding the analysis of matched-pair and stratified experimental designs in the presence of attrition. Our main objective is to clarify a number of well-known claims about the practice of dropping pairs with an attrited unit when analyzing matched-pair designs. Contradictory advice appears in the literature about whether or not dropping pairs is beneficial or harmful, and stratifying into larger groups has been recommended as a resolution to the issue. To address these claims, we derive the estimands obtained from the difference-in-means estimator in a matched-pair design both when the observations from pairs with an attrited unit are retained and when they are dropped. We find limited evidence to support the claims that dropping pairs helps recover the average treatment effect, but we find that it may potentially help in recovering a convex weighted average of conditional average treatment effects. We report similar findings for stratified designs when studying the estimands obtained from a regression of outcomes on treatment with and without strata fixed effects.


\end{abstract}
\end{spacing}



\thispagestyle{empty}
\newpage
\setcounter{page}{1}

\section{Introduction} \label{sec:intro}

In this paper we revisit some common recommendations regarding the analysis of matched-pair and stratified experimental designs in the presence of attrition. Here, we define attrition to mean that we do not observe outcomes for some subset of the experimental units. This situation may arise, for instance, if subjects refuse to participate in the experiment’s endline survey or if researchers lose track of subjects prior to observing their experimental outcomes. 

Our main objective is to clarify a number of well-known claims about the practice of dropping pairs with an attrited unit in matched-pair designs. Specifically, when one unit in a pair is lost, several contradictory suggestions have been made in the literature about whether or not experimenters should drop the remaining unit in their analyses.\footnote{Appendix \ref{sec:quotes} contains relevant excerpts from the referenced sources.} For instance, \cite{king2007politically} and \cite{bruhn2009pursuit} assert that a key advantage of matched-pair designs is that dropping pairs with an attrited unit may protect against attrition bias when attrition is a function of the matching variables. In contrast, \cite{glennerster2013running} claim that dropping pairs may \emph{increase} attrition bias, and point out that the widespread practice of including pair fixed effects in a regression of outcomes on treatment is equivalent to computing the difference-in-means estimator after dropping pairs. Accordingly, they go on to suggest that experimenters should instead stratify the units into larger groups if there is risk of attrition. \cite{donner2000design} assert that dropping pairs with an attrited unit is a \emph{requirement} in analyses of matched-pair designs with attrition, and characterize this as a weakness of matched-pair designs. As a result, they also recommend stratifying units into larger groups. 

To address these claims, we first derive the estimands obtained from the difference-in-means estimator in a matched-pair design both when the observations from pairs with an attrited unit are retained and when they are dropped. We find that the estimand produced when retaining the units is simply the difference in the mean outcomes conditional on not attriting. In contrast, the estimand produced when dropping the units is a complicated function of the mean outcomes and attrition probabilities conditional on the matching variables. Using this result, we show that dropping pairs does not recover the average treatment effect when attrition is a function of the matching variables, and instead recovers a convex weighted average\footnote{Here and throughout the paper we define a convex weighted average to be a weighted average whose coefficients are non-negative and sum to one.} of conditional average treatment effects. Moreover, we argue that natural conditions under which this convex weighted average further collapses to the average treatment effect are in fact stronger than the condition that attrition is independent of experimental outcomes. From these results we conclude that, although dropping pairs may potentially help in recovering a convex weighted average of conditional average treatment effects, we find limited evidence to support the claims that dropping pairs in a matched-pair design helps protect against attrition bias more generally.

Next, to address the claims that the issues surrounding whether or not to drop pairs with an attrited unit can be resolved by instead stratifying the experiment into larger groups, we repeat the above exercise in the context of a stratified randomized experiment where the strata are made up of a large number of observations. To mirror the analysis carried out for matched pairs, we study the estimands obtained from a regression of outcomes on treatment with and without strata fixed effects. We find analogous results: the estimand produced when omitting strata fixed effects is once again the difference in mean outcomes conditional on not attriting, and the estimand produced when including strata fixed effects is a function of the mean outcomes and attrition probabilities conditional on the strata labels with very similar properties to what was obtained for matched pairs. From these results we conclude that we do not find compelling evidence to support the idea that stratifying into larger groups resolves the issues surrounding attrition that we explore in this paper.

Including pair fixed effects when conducting inference via linear regression is a widely adopted practice \citep[see for instance the recommendations in][]{bruhn2009pursuit}, and is numerically equivalent to dropping pairs with an attrited unit. As a consequence, inference considerations sometimes drive the discussion of whether or not to drop pairs \citep[see for example Chapter 4, footnote 32 in][]{glennerster2013running}. However, in our view this should not play a primary role when deciding whether or not to drop pairs for three reasons. First, as we show in this paper, including vs. excluding pair fixed effects produces estimands with distinct interpretations in the presence of attrition. Second, as argued in \cite{bai2021inference} and \cite{bugni2018inference} (in settings without attrition), including pair/strata fixed effects is not a requirement for conducting valid inference on the ATE in matched-pair/stratified experiments, and there is no clear benefit obtained from doing so in general. Third, there are no formal results which justify the use of conventional robust standard errors in the presence of attrition (with or without fixed effects), and we conjecture that alternative inference procedures should be developed in this case (see Remark \ref{rem:inference} for a preliminary discussion). For these reasons, in this paper our primary  focus is on studying the interpretation of the resulting estimands.


 
Finally, we explore the empirical relevance of our results using experimental data collected in \cite{groh2016macroinsurance} as well as data collected from a systematic survey of all papers published in the American Economic Review (AER) and American Economic Journal: Applied Economics (AEJ: Applied) from 2020-2022 which conduct matched-pair or stratified experiments in the presence of attrition. Using these datasets we find that there can be noticeable differences between the point estimates obtained from dropping or retaining pairs with an attrited unit (or including/omitting stratum fixed effects), even when attrition is comparatively low. For instance, using the data in \cite{groh2016macroinsurance} we find an average absolute percentage difference of $13.82\%$ in point estimates across a collection of outcomes even with an average attrition rate of only $1.4\%$.

Our paper is related to a large literature on the analysis of randomized experiments with attrition. Most of this literature focuses on developing methods to recover the average treatment effect, often by either modeling the missing data process \citep{heckman1979sample,rubin2004multiple}, inverse probability weighting \citep{wooldridge2002inverse,little2019statistical}, bounding \citep{manski2000analysis,lee2009training,behaghel2015please}, or testing for the presence of attrition bias \citep{ghanem2021testing}. Instead, the focus of our paper is on studying the behavior of commonly used estimators in the analysis of matched-pair and stratified experiments. To our knowledge, the paper most similar to ours is \cite{fukumoto2022nonignorable}, who conducts finite population and super-population analyses of the bias and variance of the difference-in-means estimator in matched-pair designs with and without dropping pairs. However, his super-population analysis maintains a sampling framework where the observations are drawn together as pairs, whereas we consider a sampling framework where observations are drawn as individuals and then subsequently paired according to their covariates. As a consequence, his results and ours are not directly comparable \citep[we note that every empirical application we consider in Section \ref{sec:application} describes a specific procedure by which they stratified their sample using available covariates, and thus does not feature a sample constructed from pre-formed strata as modelled in][]{fukumoto2022nonignorable}. Moreover, \cite{fukumoto2022nonignorable} exclusively focuses on the setting of matched-pair designs and thus does not derive results for stratified randomized experiments.

The rest of the paper is structured as follows. In Section \ref{sec:setup} we describe our setup and introduce the main assumptions we consider on the attrition process. Section \ref{sec:results} presents the main results. In Section \ref{sec:application} we present an empirical illustration. Finally, we conclude in Section \ref{sec:recs} with some recommendations for empirical practice.

\section{Setup and Notation}\label{sec:setup}
Let $Y^*_i$ denote the realized outcome of interest for the $i$th unit in the absence of attrition, $D_i \in \{0,1\}$ denote treatment status for the $i$th unit and $X_i$ denote the observed, baseline covariates for the $i$th unit. Further denote by $Y_i(1)$ the potential outcome of the $i$th unit if treated and by $Y_i(0)$ the potential outcome if not treated. As usual, the realized outcome is related to the potential outcomes and treatment status by the relationship
\begin{equation}\label{eq:PO}
 Y^*_i = Y_i(1) D_i + Y_i(0) (1 - D_i)~.
\end{equation}

We consider a framework which allows for the possibility that units collected in the baseline survey may drop out (attrit) after treatment is assigned. In particular, let $R_i \in \{0, 1\}$ be an indicator where $R_i = 1$ indicates the $i$th unit is present in the endline survey (i.e. has \emph{not} attrited) and $R_i = 0$ indicates otherwise. Let $R_i(1)$ denote the potential attrition decision of the $i$th unit if treated, and $R_i(0)$ denote the potential attrition decision of the $i$th unit if not treated. As was the case for the realized outcome, the realized attrition decision is related to the potential attrition decisions and treatment status by the relationship 
\begin{equation}\label{eq:PA}
R_i = R_i(1) D_i + R_i(0) (1 - D_i)~.
\end{equation}
With these definitions in hand, we define the observed outcome to be
\begin{equation}\label{eq:OO}
Y_i = Y^*_iR_i = Y_i(1)R_i(1)D_i + Y_i(0)R_i(0)(1 - D_i)~.
\end{equation}
We note that the observed outcome is undefined if individual $i$ is not observed in the endline survey, and so we set it arbitrarily to zero in equation \eqref{eq:OO}.

We assume that we observe a sample $\{(Y_i, R_i, D_i, X_i): 1 \le i \le n\}$, obtained from i.i.d random variables $\{W_i : 1 \le i \le n\}$ where $W_i = (Y_i(1), Y_i(0), R_i(1), R_i(0), X_i)$. As a result, the distribution of the observed data is determined by (\ref{eq:PO}), (\ref{eq:PA}), (\ref{eq:OO}), $\{W_i : 1 \le i \le n\}$, and the mechanism for determining treatment assignment (which we specify in Sections \ref{sec:pairs} and \ref{sec:sfe}). We maintain the following assumption on $\{W_i: 1 \le i \le n\}$ throughout the entirety of the paper:
\begin{assumption} \label{as:Q}  
\hfill
\begin{enumerate}[(a)] 
\item $E[|Y_i(d)|] < \infty$ for $d \in \{0, 1\}$.
\item $E[R_i(d)] > 0$ for $d \in \{0, 1\}$.
\end{enumerate}
\end{assumption}
Assumption \ref{as:Q}(a) imposes mild restrictions on the moments of the potential outcomes. Assumption \ref{as:Q}(b) rules out situations where the probability of attrition is one for either treatment status. 


Our parameter of interest is the average treatment effect, denoted as
\begin{equation} \label{eq:ate-na}
\theta = E[Y_i(1) - Y_i(0)]~.
\end{equation}
Without further assumptions on the nature of attrition, $\theta$ is not  point-identified from the observed data. As a consequence, in this paper we first study the estimands produced by commonly used estimators in the analysis of matched-pair and stratified randomized experiments, and then document if and when these estimands collapse to $\theta$ under well-known, albeit strong, assumptions on the attrition process; see Remark \ref{rem:bias} for further discussion. The first assumption we consider is that attrition is independent of the potential outcomes:

\begin{assumption}\label{as:MIPO}
\[(Y_i(1), Y_i(0)) \indep (R_i(1), R_i(0))~.\]
\end{assumption}

Under Assumption \ref{as:MIPO}, the average treatment effect $\theta$ is point-identified in a classical randomized experiment by simply comparing the mean outcomes under treatment and control for the non-attritors \citep[see for instance][] {gerber2012field}. The next assumption we consider is that attrition is independent of potential outcomes conditional on some set of observable characteristics:

\begin{assumption}\label{as:MIPO|T}
For some set of observable characteristics $C_i$,
\[(Y_i(1), Y_i(0)) \indep (R_i(1), R_i(0)) \hspace{1mm} | \hspace{1mm} C_i~.\]
\end{assumption}
Although Assumption \ref{as:MIPO} does not necessarily imply Assumption \ref{as:MIPO|T} or vice versa, it is often argued that Assumption \ref{as:MIPO|T} may be easier to defend in practice \citep{moffit1999sample,hirano2001combining,gerber2012field,little2019statistical}. Under Assumption \ref{as:MIPO|T}, $\theta$ is point-identified in a classical randomized experiment by first identifying the average treatment effect conditional on each value $C = c$ and then averaging these conditional treatment effects across $C$. Note that Assumption \ref{as:MIPO|T} generalizes the assumption discussed in the introduction that attrition is a function of observable characteristics. The final assumption we consider is that attrition is independent of observable characteristics:

\begin{assumption}\label{as:indep_attrition}
For some set of observable characteristics $C_i$,
\[C_i \indep (R_i(1), R_i(0))~.\]
\end{assumption}
A useful observation for the discussion which follows is that, although Assumptions \ref{as:MIPO} and \ref{as:MIPO|T} are not nested, Assumptions \ref{as:MIPO|T} and \ref{as:indep_attrition} do in fact imply Assumption \ref{as:MIPO}. To see this, consider the following derivation:
\begin{align*}
 & P\{(Y_i(1), Y_i(0)) \in U_1, (R_i(1), R_i(0)) \in U_2\} \\ &= 
 E\left[E[I\{(Y_i(1), Y_i(0)) \in U_1\}I\{(R_i(1), R_i(0)) \in U_2\}|C_i]\right]\\
& = E\left[E[I\{(Y_i(1), Y_i(0)) \in U_1 | C_i]E[I\{(R_i(1), R_i(0)) \in U_2\}|C_i]\right] \\
& = E\left[E[I\{(Y_i(1), Y_i(0)) \in U_1 | C_i]E[I\{(R_i(1), R_i(0)) \in U_2\}]\right] \\
& = E\left[I\{(Y_i(1), Y_i(0)) \in U_1\}\right]E[I\{(R_i(1), R_i(0)) \in U_2\}] \\ 
& = P\{(Y_i(1), Y_i(0)) \in U_1\}P\{(R_i(1), R_i(0)) \in U_2\}~,
\end{align*}
where the first equality follows from the law of iterated expectations, the second equality from Assumption \ref{as:MIPO|T}, the third from Assumption \ref{as:indep_attrition}, and the fourth from the law of iterated expectations once again.




\section{Main Results}\label{sec:results}
\subsection{Matched-Pair Designs with Attrition}\label{sec:pairs}
In this section we study the estimands produced by the difference-in-means estimator in a matched-pair design when the observations from pairs with an attrited unit are retained and when they are dropped. Before defining the estimators we provide a formal description of the treatment assignment mechanism. To simplify the exposition, we assume that $n$ is even for the remainder of Section \ref{sec:pairs}. For any random variable indexed by $i$, for example $D_i$, we denote by $D^{(n)}$ the random vector $(D_1, D_2, \ldots, D_n)$. Let $\pi = \pi_n(X^{(n)})$ be a permutation of $\{1, \ldots, n\}$, potentially dependent on $X^{(n)}$. The $n/2$ matched pairs are then represented by the sets
\[ \left\{\{\pi(2j - 1), \pi(2j)\}: 1 \leq j \leq \frac{n}{2}\right\}~. \]
In other words, pairs are formed by arranging observations in the order $\{\pi(1), \pi(2), \ldots, \pi(n)\}$ according to the permutation $\pi$, and then forming pairs from the adjacent units as $\{\pi(1), \pi(2)\}$, $\{\pi(3), \pi(4)\}$, etc. Next, given such a $\pi$, we assume treatment status is assigned as follows:
\begin{assumption} \label{as:mp}
Treatment status is assigned so that
\[ (Y^{(n)}(1), Y^{(n)}(0), R^{(n)}(1), R^{(n)}(0)) \indep D^{(n)} | X^{(n)} \]
and, conditional on $X^{(n)}$, $(D_{\pi(2j - 1)}, D_{\pi(2j)}), 1 \leq j \leq n/2$ are i.i.d.\ and each uniformly distributed over $\{(0, 1), (1, 0)\}$.
\end{assumption}
To summarize, the assignment mechanism first forms pairs of units (according to $\pi$) and then assigns both treatments exactly once in each pair at random. The first estimator we consider is the standard difference-in-means estimator computed on non-attritors:
\begin{equation} \label{eq:est}
\hat \theta_n = \frac{\sum_{1 \leq i \leq n} Y_i R_i D_i}{\sum_{1 \leq i \leq n} R_i D_i} - \frac{\sum_{1 \leq i \leq n} Y_i R_i (1 - D_i)}{\sum_{1 \leq i \leq n} R_i (1 - D_i)}~.
\end{equation}
Note that $\hat{\theta}_n$ may be obtained as the estimator of the coefficient on $D_i$ in an ordinary least squares regression of $Y_i$ on a constant and $D_i$, computed on the non-attritors. The second estimator we consider is the difference-in-means estimator computed by first dropping any observations belonging to a pair with an attritor:
\begin{multline*}
\hat \theta_n^{\rm drop} = \left ( \sum_{1 \leq j \leq n/2} R_{\pi(2j - 1)} R_{\pi(2j)} \right )^{-1} \\
\times \left( \sum_{1 \leq j \leq n/2} R_{\pi(2j - 1)} R_{\pi(2j)} (Y_{\pi(2j - 1)} - Y_{\pi(2j)}) (D_{\pi(2j - 1)} - D_{\pi(2j)}) \right)~.
\end{multline*}
Note that $\hat{\theta}_n^{\rm drop}$ corresponds to the estimator recommended in \cite{bruhn2009pursuit} and \cite{king2007politically}. We emphasize that, in the absence of attrition, $\hat{\theta}_n$ and $\hat{\theta}^{\rm drop}_n$ are numerically equivalent. 

As a consequence of the Frisch-Waugh-Lovell theorem, $\hat{\theta}_n^{\rm drop}$ can equivalently be obtained as the ordinary least squares estimator of the coefficient on $D_i$ in the linear regression of $Y_i$ on $D_i$ and pair fixed effects computed on the non-attritors (i.e. individuals with $R_i = 1$):\footnote{See Appendix \ref{sec:FWL} for a derivation of this fact.}
\begin{equation} \label{eq:pfe}
Y_i = \theta^{\rm drop} D_i + \sum_{1 \leq j \leq n/2} \delta_j I \{i \in \{\pi(2j - 1), \pi(2j)\}\} + \epsilon_i \hspace{3mm} \text{(for individuals with $R_i = 1$)}~.
\end{equation}
Similar regression specifications are extremely common in the analysis of matched-pair experiments. See, for example, \cite{ashraf2006deposit}, \cite{angrist2009effects}, \cite{crepon2015estimating}, \cite{bruhn2016impact}, and \cite{fryer2018pupil}. 

We impose the following assumption in addition to Assumption \ref{as:Q}:
\begin{assumption} \label{as:Q-lip}
\hfill
\begin{enumerate}[(a)]
\item $E[R_i(d) | X_i = x]$ is Lipschitz in $x$ for $d \in \{0, 1\}$.
\item $E[Y_i(d) R_i(d) | X_i = x]$ is Lipschitz in $x$ for $d \in \mathcal \{0, 1\}$.
\end{enumerate}
\end{assumption}
\noindent Assumptions \ref{as:Q-lip}(a)--(b) are smoothness requirements that ensure that units that are ``close'' in terms of their baseline covariates are also ``close'' in terms of their potential attrition indicators and potential outcomes on average. Similar smoothness requirements are also imposed in \cite{bai2021inference} and \cite{bai2022optimality}.

Finally, we require that the matched-pair design is such that the units in each pair are ``close'' in terms of their baseline covariates in the following sense:
\begin{assumption} \label{as:close}
The pairs used in determining treatment status satisfy
\[ \frac{1}{n} \sum_{1 \leq j \leq n} \|X_{\pi(2j -1)} - X_{\pi(2j)}\| \stackrel{P}{\to} 0~. \]
\end{assumption}
See \cite{bai2021inference} for sufficient conditions for Assumption \ref{as:close}. In particular, if $\mathrm{dim}(X_i) = 1$, then Assumption \ref{as:close} is satisfied if $E[|X_i|] < \infty$ and we construct pairs by simply ordering the units from smallest to largest according to $X_i$ and then pairing adjacent units. For the case $\mathrm{dim}(X_i) > 1$, \cite{bai2021inference} provide sufficient conditions under which Assumption \ref{as:close} is satisfied when using the popular R package {\tt nbpMatching}. Using appropriate laws of large numbers developed in \cite{bai2021inference}, we now establish the following result:
\begin{theorem} \label{thm:pair}
Suppose the data satisfy Assumptions \ref{as:Q} and \ref{as:Q-lip} and the treatment assignment mechanism satisfies Assumptions \ref{as:mp} and \ref{as:close}. Then, as $n \to \infty$, $\hat \theta_n  \stackrel{P}{\to}  \theta^{\rm obs}$, where
\[\theta^{\rm obs} = \frac{E[R_i(1) Y_i(1)]}{E[R_i(1)]} - \frac{E[R_i(0) Y_i(0)]}{E[R_i(0)]} = E[Y_i(1)|R_i(1) = 1] - E[Y_i(0)|R_i(0) = 1]~, \]
and $\hat \theta_n^{\rm drop} \stackrel{P}{\to} \theta^{\rm drop}$, where
\[\theta^{\rm drop} = E[\tau^{\rm obs}(X_i) \rho(X_i)]~,\] with
 \begin{align*}
 \tau^{\rm obs}(x) & = E[Y_i(1) | R_i(1) = 1, X_i = x] - E[Y_i(0) | R_i(0) = 1, X_i = x] \\
 \rho(x) & = \frac{E[R_i(0)|X_i=x]E[R_i(1)|X_i=x]}{E[E[R_i(0)|X_i]E[R_i(1)|X_i]]}~.
 \end{align*}
\end{theorem}

Theorem \ref{thm:pair} shows that the estimand produced by the difference-in-means estimator, $\theta^{\rm obs}$, is simply the difference in the mean outcomes conditional on not attriting (under the additional assumption that $R_i(1) = R_i(0)$ this could be interpreted as the average treatment effect for units who do not attrit: see Remark \ref{rem:attrit_units} for details). It follows immediately that, under Assumption \ref{as:MIPO}, $\theta^{\rm obs} = \theta$ and thus under this assumption we recover the average treatment effect.

On the other hand, the estimand produced by first dropping units belonging to a pair with an attritor, $\theta^{\rm drop}$, is a complicated function of the mean outcomes and attrition probabilities conditional on the matching variables. First, note that unlike $\theta^{\rm obs}$, $\theta^{\rm drop}$ does not collapse to $\theta$ under Assumption \ref{as:MIPO}. Moreover, $\theta^{\rm drop}$ does not collapse to $\theta$ under Assumption \ref{as:MIPO|T} with $C_i = X_i$ either. Instead, under Assumption \ref{as:MIPO|T} with $C_i = X_i$, $\tau^{\rm obs}(x) = \tau(x)$ where $\tau(x) = E[Y_i(1) - Y_i(0)|X_i=x]$, so that
\[\theta^{\rm drop} = E[\tau(X_i)\rho(X_i)]~,\]
i.e. $\theta^{\rm drop}$ may be written as a convex weighted average of the conditional average treatment effects $\tau(x)$. In some special cases this convex-weighted average has a simple and transparent interpretation: consider for example a setting where $X_i$ is a binary variable, and suppose that attrition is such that units with $X_i = 1$ always appear in the endline survey, so that $R_i(1) = R_i(0) = 1$ if $X_i = 1$, but units with $X_i = 0$ appear only if they are treated, so that $R_i(1) = 1$ and $R_i(0) = 0$ if $X_i = 0$. Then,
\[ \rho(1) = \frac{1}{P \{X_i = 1\}}~, \] 
and $\rho(0) = 0$. We thus have that in this case
\[ \theta^{\rm drop} = E[Y_i(1) - Y_i(0) | X_i = 1]~, \]
which is the average treatment effect for those units with $X_i = 1$. In contrast, $\theta^{\rm obs}$ does not lend itself to a straightforward causal interpretation in this example (however in Remark \ref{rem:attrit_units} we provide a favorable interpretation of $\theta^{\rm obs}$ under Assumption \ref{as:MIPO|T} and the additional assumption that $R_i(1) = R_i(0)$).

In general, straightforward algebra shows that $\rho(x) = 1$ if and only if
\begin{equation}\label{eq:proportion}
E[R_i(0)|X_i = x] = \frac{E[E[R_i(0)|X_i]E[R_i(1)|X_i]]}{E[R_i(1)|X_i = x]}~.
\end{equation}
In words, $\rho(x) = 1$ if and only if the conditional probability of attrition under treatment is inversely proportional to the conditional probability of attrition under control. A natural assumption which guarantees (\ref{eq:proportion}) for all $x$ is Assumption \ref{as:indep_attrition} with $C_i = X_i$, so that attrition is independent of the matching variables $X_i$. Finally, we note that under Assumption \ref{as:indep_attrition} with $C_i = X_i$, it follows that $\theta^{\rm drop} = \theta^{\rm obs}$. As a result, $\theta^{\rm drop} = \theta$ under Assumptions \ref{as:MIPO} and \ref{as:indep_attrition}. We summarize the above discussion in the following corollary: 
 
\begin{corollary}\label{cor:drop_recover}
\hfill
\begin{enumerate}[(a)]
\item Under Assumption \ref{as:MIPO}, $\theta^{\rm obs} = \theta$. 
\item Under Assumption \ref{as:MIPO|T} with $C_i = X_i$, $\theta^{\rm drop} = E[\tau(X_i)\rho(X_i)]$. 
\item Under Assumption \ref{as:indep_attrition} with $C_i = X_i$, $\theta^{\rm drop} = \theta^{\rm obs}$. 
\item Under Assumption \ref{as:indep_attrition} with $C_i = X_i$ and either Assumption \ref{as:MIPO} or Assumption \ref{as:MIPO|T} with $C_i = X_i$, $\theta^{\rm drop} = \theta$.
\end{enumerate}
\end{corollary}

We conclude this section by noting that, as explained in the derivation following the statement of Assumption \ref{as:indep_attrition}, Assumptions \ref{as:MIPO|T} and \ref{as:indep_attrition} \emph{imply} Assumption \ref{as:MIPO}.  In other words, we see that the sufficient conditions provided in Corollary \ref{cor:drop_recover} under which $\theta^{\rm drop} = \theta$ are in fact \emph{stronger} than the conditions required for $\theta^{\rm obs} = \theta$. We thus find limited evidence to support the claims that dropping pairs in a matched-pair design helps in reducing attrition bias. However, we emphasize that dropping pairs may potentially help in recovering a convex weighted average of conditional average treatment effects. 

\begin{remark}\label{rem:inference}
In Appendix \ref{sec:norm_theta} we develop the requisite distributional results to use $\hat{\theta}_n$ for inference about $\theta^{\rm obs}$. In contrast, the large sample distribution of $\hat{\theta}^{\rm drop}_n$ seems non-trivial to characterize and may in fact feature an asymptotic bias in general. For this reason we leave an in-depth study of the limiting distribution of $\hat{\theta}^{\rm drop}_n$ to future work.
\end{remark}

\begin{remark}\label{rem:bias}
We note that, in the absence of additional assumptions like Assumptions \ref{as:MIPO}--\ref{as:indep_attrition}, we are not able to conclude that either $\theta^{\rm obs}$ or $\theta^{\rm drop}$ is less biased for $\theta$ relative to the other, and in fact it is possible to construct data generating processes where either estimand is closer to the true average treatment effect. We present a concrete construction of such a set of DGPs in Appendix \ref{sec:numerical}.
\end{remark}

\begin{remark}\label{rem:attrit_units}
From Theorem \ref{thm:pair} we also observe that, in the absence of additional assumptions like Assumptions \ref{as:MIPO}--\ref{as:indep_attrition}, neither $\theta^{\rm obs}$ nor $\theta^{\rm drop}$ can be interpreted as a ``treatment effect parameter" (in the sense that neither parameter can be interpreted as an average treatment effect for some subset of individuals or more generally as a weighted average of treatment effects). This is because the subgroup of units who attrit under treatment ($R_i(1) = 0$) may not correspond to the subgroup of units who attrit under control ($R_i(0) = 0$). Under the additional assumption that $R_i(1) = R_i(0)$, so that these subgroups coincide, we obtain
\[\theta^{\rm obs} = E[Y_i(1) - Y_i(0) | R_i = 1]~,\] 
and then $\theta^{\rm obs}$ could be understood as the average treatment effect for units who do not attrit. Imposing the same assumption for $\theta^{\rm drop}$ we obtain that
\[\theta^{\rm drop} = \frac{E[\left(Y_i(1) - Y_i(0)\right) R_i E[R_i | X_i]]}{E[R_iE[R_i | X_i]]}~,\]
and then $\theta^{\rm drop}$ could be understood as a ``probability of attrition"-weighted average of individual-level treatment effects for the non-attritors. If we additionally impose Assumption \ref{as:MIPO|T}, we alternatively obtain
\begin{equation*}
    \theta^{\rm obs} = \frac{E[\tau(X_i)E[R_i | X_i]]}{P(R_i=1)}, \quad \theta^{\rm drop} = \frac{E[\tau(X_i)E[R_i | X_i]^2]}{E[E[R_i | X_i]^2]}~.
\end{equation*}
In this case, \emph{both} parameters can be interpreted as convex weighted averages of conditional average treatment effects, with the main difference being that $\theta^{\rm obs}$ is weighted using the the conditional attrition rate $E[R_i | X_i]$, whereas $\theta^{\rm drop}$  ``doubles down" by weighting using the squared conditional attrition rate $E[R_i | X_i]^2$.
\end{remark}

\begin{remark} \label{rem:bfe}
Regardless of whether or not a practitioner finds the interpretation of $\theta^{\rm drop}$ more or less attractive than the interpretation of $\theta^{\rm obs}$, it is crucial to note that, \emph{even in the absence of attrition}, inferences produced using robust standard errors obtained from a regression with pair fixed effects are generally conservative, but in some cases may in fact be \emph{invalid}, in the sense that the limiting rejection probability could be strictly larger than the nominal level. See \cite{bai2022inference} and \cite{de2020level} for details. 
\end{remark}


\subsection{Stratified Designs with Attrition}\label{sec:sfe}
In this section we repeat the exercise presented in Section \ref{sec:pairs} but in the context of stratified designs. Before describing the estimators, we provide a description of the class of treatment assignment mechanisms we consider. In words, our results accommodate any treatment assignment mechanism which first partitions the covariate space into a finite number of ``large'' strata, and then performs treatment assignment independently across strata so as to achieve ``balance'' within each stratum. Formally, let $S:\text{supp}(X_i) \rightarrow \mathcal{S}$ be a function which maps the support of the covariates into a finite set $\mathcal{S}$ of strata labels. For $1 \le i \le n$, let $S_i = S(X_i)$ denote the strata label of individual $i$. For $s \in \mathcal{S}$, let 
\[D_n(s) = \sum_{1 \le i \le n}(D_i - \nu)I\{S_i = s\}~,\]
where $\nu \in (0, 1)$ denotes the ``target'' proportion of units to assign to treatment in each stratum. Intuitively, $D_n(s)$ measures the amount of imbalance in stratum $s$ relative to the target proportion $\nu$. Our requirements on the treatment assignment mechanism can then be summarized as follows: 

\begin{assumption} \label{as:strat}
The treatment assignment mechanism is such that 
\begin{enumerate}[(a)]
\item $W^{(n)} \indep D^{(n)} | S^{(n)}$.
\item $\frac{D_n(s)}{n} \stackrel{P}{\to} 0$ for every $s \in \mathcal{S}$.
\end{enumerate}
\end{assumption}

Assumption \ref{as:strat}(a) simply requires that treatment assignment be exogenous conditional on the strata labels. Assumption \ref{as:strat}(b) formalizes the requirement that the assignment mechanism performs treatment assignment so as to achieve ``balance'' within strata. Assumption \ref{as:strat}(b) is a relatively mild assumption which is satisfied by most stratified randomization procedures employed in field experiments: see \cite{bugni2018inference} for examples.

As before, the first estimator we consider is the standard difference-in-means estimator computed on non-attritors $\hat{\theta}_n$. The second estimator we consider, denoted $\hat{\theta}_n^{\rm sfe}$, is the estimator obtained as the estimator of the coefficient on $D_i$ in an ordinary least squares regression of $Y_i$ on $D_i$ and strata fixed effects computed on the non-attritors:
\[Y_i = \theta^{\rm sfe}D_i + \sum_{s \in \mathcal{S}}\delta_sI\{S_i = s\} + \epsilon_i \hspace{3mm} \text{(for individuals with $R_i = 1$)}~.\]
Similar regression specifications are extremely common in the analysis of stratified randomized experiments. See, for example, \cite{bruhn2009pursuit}, \cite{duflo2015education}, \cite{glennerster2013running}, \cite{de_mel2019labor}, and \cite{callen2020data}. Using appropriate laws of large numbers developed in \cite{bugni2018inference}, we now establish the following result:
\begin{theorem}\label{thm:covariate-adaptive}
Suppose the data satisfy Assumptions \ref{as:Q} and the treatment assignment mechanism satisfies Assumption \ref{as:strat}, Then as $n \to \infty$, $\hat{\theta}_n \stackrel{P}{\to} \theta^{\rm obs}$, where
\[\theta^{\rm obs} = \frac{E[R_i(1) Y_i(1)]}{E[R_i(1)]} - \frac{E[R_i(0) Y_i(0)]}{E[R_i(0)]} = E[Y_i(1)|R_i(1) = 1] - E[Y_i(0)|R_i(0) = 1]~,\]
and $\hat{\theta}^{\rm sfe}_n \stackrel{P}{\to} \theta^{\rm sfe}$, where
\begin{multline*}
\theta^{\rm sfe} = \left ( E \left [ \frac{E[R_i(1) | S_i] E[R_i(0) | S_i]}{\nu E[R_i(1) | S_i] + (1 - \nu) E[R_i(0) | S_i]} \right ] \right )^{-1} \\
\times
E \left [ \frac{E[R_i(1) Y_i(1) | S_i] E[R_i(0) | S_i] - E[R_i(0) Y_i(0) | S_i] E[R_i(1) | S_i]}{\nu E[R_i(1) | S_i] + (1 - \nu) E[R_i(0) | S_i]} \right ]~.    
\end{multline*}

\end{theorem}
The conclusions we draw from Theorem \ref{thm:covariate-adaptive} closely mirror those of Theorem \ref{thm:pair}. In this case, under Assumption \ref{as:MIPO|T} with $C_i = S_i$, 
\[\theta^{\rm sfe} = E\left[\tau(S_i)\lambda(S_i)\right]~,\]
where $\tau(s) = E[Y_i(1) - Y_i(0)|S_i = s]$ and
\[\lambda(s) =  \left ( E \left [ \frac{E[R_i(1) | S_i] E[R_i(0) | S_i]}{\nu E[R_i(1) | S_i] + (1 - \nu) E[R_i(0) | S_i]} \right ] \right )^{-1}\times\frac{E[R_i(1)| S_i = s] E[R_i(0) | S_i = s]}{\nu E[R_i(1) | S_i=s] + (1 - \nu) E[R_i(0) | S_i=s]}~,\]
so that $\theta^{\rm sfe}$ is also a convex weighted average of the strata-level treatment effects $\tau(s)$, although the weights $\lambda(s)$ are arguably more complicated to interpret than the weights $\rho(x)$ defined in Section \ref{sec:pairs}. Straightforward algebra shows that $\lambda(s) = 1$ if and only if
\begin{equation}\label{eq:strat_proportion}
E[R_i(1)|S_i = s] = \frac{E[R_i(0)|S_i = s](1 - \nu)\Lambda}{E[R_i(0)|S_i = s] - \Lambda\nu}~,
\end{equation}
where $\Lambda = E \left [ \frac{E[R_i(1) | S_i] E[R_i(0) | S_i]}{\nu E[R_i(1) | S_i] + (1 - \nu) E[R_i(0) | S_i]} \right ]$. Conditions under which this holds seem difficult to articulate in words, but once again a natural assumption which guarantees (\ref{eq:strat_proportion}) for every $s \in \mathcal{S}$ is that Assumption \ref{as:indep_attrition} is satisfied with $C_i = S_i$. We summarize these observations in the following corollary:

\begin{corollary}\label{cor:strat_recover}
\hfill
\begin{enumerate}[(a)]
\item Under Assumption \ref{as:MIPO}, $\theta^{\rm obs} = \theta$. 
\item Under Assumption \ref{as:MIPO|T} with $C_i = S_i$, $\theta^{\rm sfe} = E[\tau(S_i)\lambda(S_i)]$. 
\item Under Assumption \ref{as:indep_attrition} with $C_i = S_i$, $\theta^{\rm sfe} = \theta^{\rm obs}$.
\item Under Assumption \ref{as:indep_attrition} with $C_i = S_i$ and either Assumption \ref{as:MIPO} or Assumption \ref{as:MIPO|T} with $C_i = S_i$, we obtain $\theta^{\rm sfe} = \theta$.
\end{enumerate}
\end{corollary}

We conclude this section by stating that, given how closely the results presented in Section \ref{sec:sfe} mirror those in Section \ref{sec:pairs}, we do not find compelling evidence to support the idea that stratifying into larger groups resolves the issues surrounding attrition that we explore in this paper.




\section{Empirical Illustrations}\label{sec:application}

\subsection{Re-analysis of \cite{groh2016macroinsurance}}
In this section we illustrate the potential empirical relevance of deciding whether or not to drop pairs with an attrited unit using the experimental data collected in \cite{groh2016macroinsurance}, which implemented a matched-pair design in the presence of attrition. The regression specifications in the paper contain pair fixed effects, which, as explained in Section \ref{sec:pairs}, is mechanically equivalent to dropping pairs with an attrited unit when regressing outcomes on a constant and treatment. 


\cite{groh2016macroinsurance} study the effect of insuring microenterprises (clients) against macroeconomic instability and political uncertainty in post-revolution Egypt. A baseline survey was completed for 2961 clients, who were then randomly assigned to treatment (1481 individuals) and control (1480 individuals) using a matched-pair design\footnote{Per the authors, they ``created matched pairs [...] to minimize the Mahalanobis distance between the values of 13 variables that [they] hypothesized may determine loan take-up and investment decisions''. The final assignment contained \textit{one} stratum with 16 individuals, each belonging to a different branch office. We follow the authors' methodology in \textit{keeping} this stratum when we conduct our analysis in Table \ref{table:est-main-GM}. We drop these when we perform additional analyses in Table \ref{table:est-supp-GM}.}. 
In Table \ref{table:est-main-GM} we reproduce the intention-to-treat estimates from Table 7 of their paper, which presents estimated treatment effects on profits, revenues, employees and household consumption. ``Original'' corresponds to the estimates obtained from running the regression specifications in the original paper which include pair fixed effects, and $\hat{\theta}_n$ corresponds to estimates obtained from running an identical regression specification without pair fixed effects (we note that we were able successfully reproduce all of the reported estimates from the paper). We find an average absolute percentage difference of $13.82\%$\footnote{Here the absolute percentage difference is computed as $\left(\frac{|\text{Original} - \hat{\theta}_n|}{|\text{Original}|}\right)\times 100$.} for the point estimates of these effects, with the largest differences appearing for profits and revenue.

\begin{table}[htbp]
  \centering
  \caption{Summary of Estimates Obtained from Empirical Application: \cite{groh2016macroinsurance}}
  \begin{adjustbox}{max width=\linewidth,center}
    \begin{tabular}{ccccccccc}
    \toprule
          &       & High  &       & High  & Number & Any   & Owner's & Monthly \\
          & Profits & Profit & Revenue & Revenue & Employees & Worker & Hours & Consumption \\
    \midrule
   Original & -59.702 & -0.009 & -737.199 & -0.020 & -0.024 & 0.008 & -0.655 & -7.551 \\ 
   \addlinespace
   $\widehat{\theta}_n$ & -38.642 & -0.007 & -692.818 & -0.020 & -0.023 & 0.007 & -0.773 & -7.337 \\ 
   \addlinespace
Attrition (\%) & 2.086 & 2.086 & 2.153 & 2.153 & 1.783 & 1.783 & 1.480 & 0.000 \\ 
    \bottomrule
    \end{tabular}%
    \end{adjustbox}
\begin{tablenotes} \footnotesize 
\item Note: For each outcome listed in Table 7 of \cite{groh2016macroinsurance}, we report (a) the original estimates obtained in paper (``Original''), (b) the estimate on treatment status without pair fixed effects ($\widehat{\theta}_n$), and (c) the attrition rate in \% by outcome, defined as [number of individuals with missing outcome / total number of individuals]. The regression specifications here include baseline covariates; see Table \ref{table:est-supp-GM} for analogous results without baseline covariates included.  
\end{tablenotes}
\label{table:est-main-GM}\end{table}%

One caveat to the findings in Table \ref{table:est-main-GM} is that the setting does not map exactly into our theoretical results: first, both regressions control for baseline covariates and second, the final assignment contained \emph{one} stratum with 16 individuals, each belonging to a different branch office. Given this, in Table \ref{table:est-supp-GM} we report the intention-to-treat estimates without baseline covariates and without this additional stratum. In this case we find an average absolute percentage difference of $15.61\%$ for the point estimates of the effects. We emphasize that we consider these difference particularly salient given that attrition is quite low (on average $1.4\%$ across the outcomes), and that in the absence of attrition these estimates would be \emph{numerically identical}, as illustrated from the estimates of the effect of treatment for monthly consumption.

\begin{table}[ht!]
  \centering
  \caption{Summary of Additional Estimates Obtained from Empirical Application: \cite{groh2016macroinsurance}}
  \begin{adjustbox}{max width=\linewidth,center}
    \begin{tabular}{ccccccccc}
    \toprule
   &       & High  &       & High  & Number & Any   & Owner's & Monthly \\
          & Profits & Profit & Revenue & Revenue & Employees & Worker & Hours & Consumption \\
    \midrule
   Original & -91.197 & -0.011 & -967.967 & -0.024 & -0.032 & 0.004 & -0.561 & -3.600 \\ 
   \addlinespace
   $\hat{\theta}_n$ & -80.058 & -0.009 & -888.608 & -0.023 & -0.026 & 0.005 & -0.481 & -3.600 \\
   \addlinespace
   Attrition (\%) & 1.755 & 1.755 & 1.824 & 1.824 & 1.411 & 1.411 & 1.514 & 0.000 \\
    \bottomrule
    \end{tabular}%
    \end{adjustbox}
\begin{tablenotes} \footnotesize 
\item Note: For each outcome regression specification listed in Table 7 of \cite{groh2016macroinsurance}, we report (a) the original estimates obtained in paper (``Original''), (b) the estimate on treatment status without pair fixed effects ($\hat{\theta}_n$), and (c) the attrition rate in \% by outcome, defined as [number of individuals with missing outcome / total number of individuals]. The regression specifications here exclude baseline covariates from the authors' original work. 
\end{tablenotes}
\label{table:est-supp-GM}\end{table}%


\subsection{Re-analysis of Recent Publications in AER \& AEJ: Applied} \label{sec:application-aer-aej}

Next, we perform a similar exercise using the data from a systematic survey of all papers published in the American Economic Review (AER) and the American Econonomic Journal: Applied Economics (AEJ: Applied) from 2020-2022 which conducted matched-pair or stratified randomized experiments in the presence of attrition. Our survey identified seven such papers: \cite{abebe2021selection}, \cite{attanasio2020estimating}, \cite{carter2021subsidies}, \cite{casaburi2021using}, \cite{dhar2022reshaping}, \cite{hjort2021research}, and \cite{romero2020outsourcing}. For each paper, we collected a set of ``relevant" regression specifications,\footnote{We note that in some papers such as \cite{attanasio2020estimating} and \cite{casaburi2021using} the primary results were not necessarily the output of a linear regression, and so in these cases we selected a collection of preliminary regression analyses. In other papers such as \cite{hjort2021research}, the primary results were LATE estimates obtained via IV regression, and so in these cases we report the intention to treat analyses. Specific selection details for each paper are outlined in Appendix \ref{sec:addl-info-empirical}.} and reproduced these regressions with and without pair/stratum fixed effects (we note that we were able to successfully reproduce all of the reported estimates from each paper). In Figure \ref{fig:comparison} we report the average absolute percentage change (computed as $\left(\frac{|\text{Alternative} - \text{Original}|}{|\text{Original}|}\right)\times 100$, where ``Original" corresponds to the point estimate computed in the paper, and ``Alternative" corresponds to the estimate computed from the alternative specification with or without fixed effects) across all specifications for each paper. Similar to our findings for \cite{groh2016macroinsurance}, we find that there can be noticeable differences in the point estimates with and without fixed effects (although we emphasize that we do not claim that these differences are necessarily statistically significant).

\begin{figure}[ht!]
    \centering
    \includegraphics[width=\textwidth]{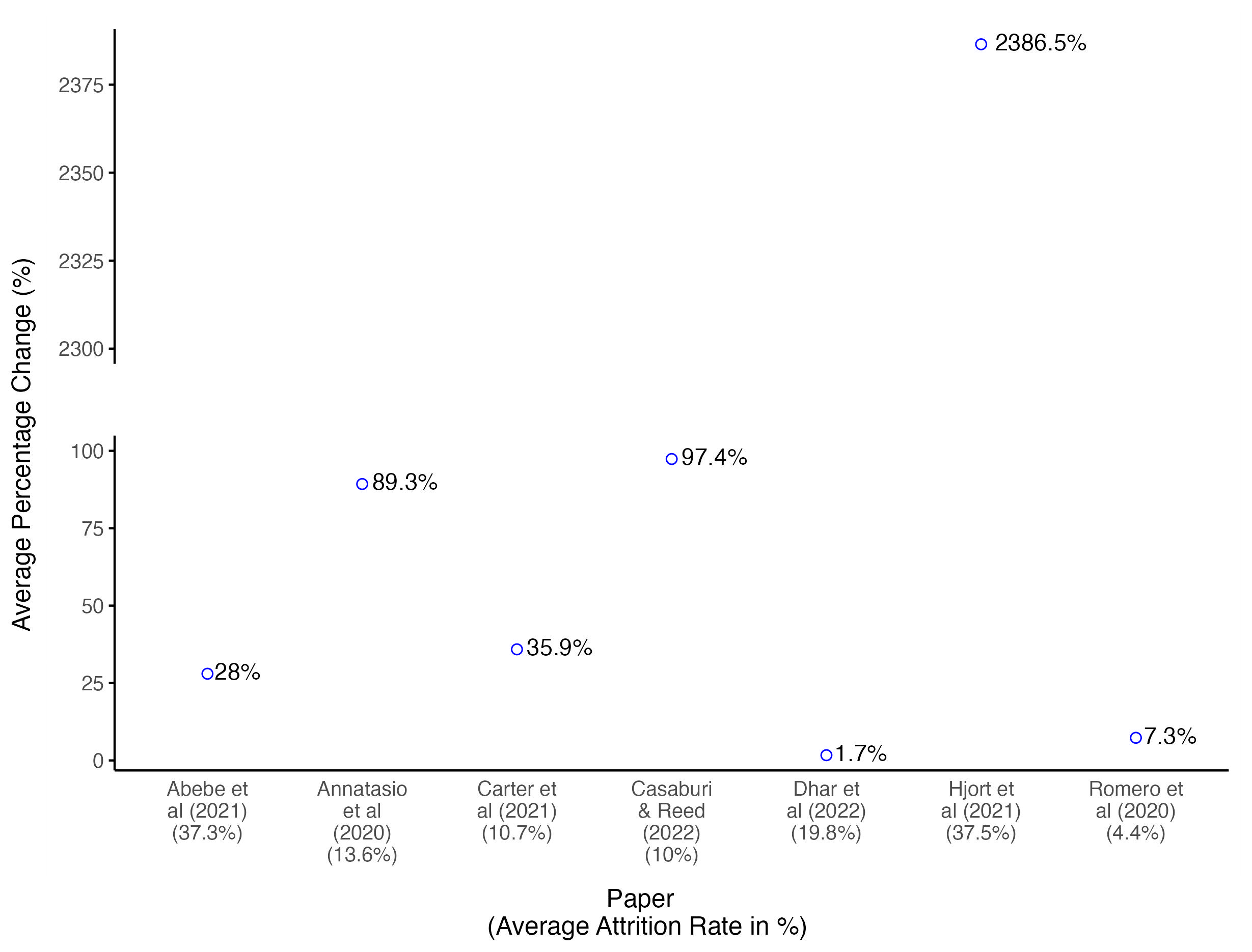}
    \caption{Average absolute percentage difference for ``Original" vs ``Alternative" point estimates. Average attrition rate, defined as [number of individuals with missing outcome / total number of individuals] is reported in parentheses below each author label.}
    \label{fig:comparison}
\end{figure}

\section{Recommendations for Empirical Practice}\label{sec:recs}
We conclude with some recommendations for empirical practice based on our theoretical results. Our main takeaway is that choosing whether or not to include pair/strata fixed effects when attrition is a concern can make a substantive difference to empirical findings and to the interpretation of the resulting estimand. In our view, unless practitioners are interested in recovering the convex-weighted averages produced by $\theta^{\rm drop}$ and $\theta^{\rm sfe}$ under a conditional independence assumption (Assumption \ref{as:MIPO|T}), primary analyses should be based on regressions \emph{without} pair/strata fixed effects: the resulting estimand $\theta^{\rm obs}$ has a simple interpretation in the absence of any assumptions, and collapses to the average treatment effect under arguably weaker assumptions than $\theta^{\rm drop}$ and $\theta^{\rm sfe}$. A secondary benefit of $\theta^{\rm obs}$ is that, under the additional assumption that $R_i(1) = R_i(0)$, $\theta^{\rm obs}$ \emph{also} enjoys an interpretation as a convex-weighted average under Assumption \ref{as:MIPO|T}, with weights which may be more desirable than those appearing in $\theta^{\rm drop}$ or $\theta^{\rm sfe}$ in that they do not ``double-up" on attrition: see Remark \ref{rem:attrit_units} for details.

\clearpage
\bibliography{attrition}

\begin{thebibliography}{39}
\newcommand{\enquote}[1]{``#1''}
\expandafter\ifx\csname natexlab\endcsname\relax\def\natexlab#1{#1}\fi

\bibitem[\protect\citeauthoryear{Abebe, Caria, and Ortiz-Ospina}{Abebe
  et~al.}{2021}]{abebe2021selection}
\textsc{Abebe, G., A.~S. Caria, and E.~Ortiz-Ospina} (2021): \enquote{The
  selection of talent: Experimental and structural evidence from ethiopia,}
  \emph{American Economic Review}, 111, 1757--1806.

\bibitem[\protect\citeauthoryear{Angrist and Lavy}{Angrist and
  Lavy}{2009}]{angrist2009effects}
\textsc{Angrist, J. and V.~Lavy} (2009): \enquote{The {Effects} of {High}
  {Stakes} {High} {School} {Achievement} {Awards}: {Evidence} from a
  {Randomized} {Trial},} \emph{American Economic Review}, 99, 1384--1414.

\bibitem[\protect\citeauthoryear{Ashraf, Karlan, and Yin}{Ashraf
  et~al.}{2006}]{ashraf2006deposit}
\textsc{Ashraf, N., D.~Karlan, and W.~Yin} (2006): \enquote{Deposit
  {Collectors},} \emph{Advances in Economic Analysis \& Policy}, 5.

\bibitem[\protect\citeauthoryear{Attanasio, Cattan, Fitzsimons, Meghir, and
  Rubio-Codina}{Attanasio et~al.}{2020}]{attanasio2020estimating}
\textsc{Attanasio, O., S.~Cattan, E.~Fitzsimons, C.~Meghir, and
  M.~Rubio-Codina} (2020): \enquote{Estimating the production function for
  human capital: results from a randomized controlled trial in Colombia,}
  \emph{American Economic Review}, 110, 48--85.

\bibitem[\protect\citeauthoryear{Bai}{Bai}{2022}]{bai2022optimality}
\textsc{Bai, Y.} (2022): \enquote{Optimality of {Matched}-{Pair} {Designs} in
  {Randomized} {Controlled} {Trials},} Tech. Rep. arXiv:2206.07845, arXiv,
  arXiv:2206.07845 [econ, math, stat] type: article.

\bibitem[\protect\citeauthoryear{Bai, Liu, Shaikh, and Tabord-Meehan}{Bai
  et~al.}{2023}]{bai2023inference}
\textsc{Bai, Y., J.~Liu, A.~M. Shaikh, and M.~Tabord-Meehan} (2023):
  \enquote{Inference in {Cluster} {Randomized} {Trials} with {Matched}
  {Pairs},} ArXiv:2211.14903 [econ, stat].

\bibitem[\protect\citeauthoryear{Bai, Liu, and Tabord-Meehan}{Bai
  et~al.}{2022}]{bai2022inference}
\textsc{Bai, Y., J.~Liu, and M.~Tabord-Meehan} (2022): \enquote{Inference for
  Matched Tuples and Fully Blocked Factorial Designs,} \emph{arXiv preprint
  arXiv:2206.04157}.

\bibitem[\protect\citeauthoryear{Bai, Romano, and Shaikh}{Bai
  et~al.}{2021}]{bai2021inference}
\textsc{Bai, Y., J.~P. Romano, and A.~M. Shaikh} (2021): \enquote{Inference in
  {Experiments} with {Matched} {Pairs}*,} \emph{Journal of the American
  Statistical Association}, 0, 1--37, publisher: Taylor \& Francis \_eprint:
  https://doi.org/10.1080/01621459.2021.1883437.

\bibitem[\protect\citeauthoryear{Behaghel, Cr{\'e}pon, Gurgand, and
  Le~Barbanchon}{Behaghel et~al.}{2015}]{behaghel2015please}
\textsc{Behaghel, L., B.~Cr{\'e}pon, M.~Gurgand, and T.~Le~Barbanchon} (2015):
  \enquote{Please call again: Correcting nonresponse bias in treatment effect
  models,} \emph{Review of Economics and Statistics}, 97, 1070--1080.

\bibitem[\protect\citeauthoryear{Bruhn, Leão, Legovini, Marchetti, and
  Zia}{Bruhn et~al.}{2016}]{bruhn2016impact}
\textsc{Bruhn, M., L.~d.~S. Leão, A.~Legovini, R.~Marchetti, and B.~Zia}
  (2016): \enquote{The {Impact} of {High} {School} {Financial} {Education}:
  {Evidence} from a {Large}-{Scale} {Evaluation} in {Brazil},} \emph{American
  Economic Journal: Applied Economics}, 8, 256--295.

\bibitem[\protect\citeauthoryear{Bruhn and McKenzie}{Bruhn and
  McKenzie}{2009}]{bruhn2009pursuit}
\textsc{Bruhn, M. and D.~McKenzie} (2009): \enquote{In {Pursuit} of {Balance}:
  {Randomization} in {Practice} in {Development} {Field} {Experiments},}
  \emph{American Economic Journal: Applied Economics}, 1, 200--232.

\bibitem[\protect\citeauthoryear{Bugni, Canay, and Shaikh}{Bugni
  et~al.}{2018}]{bugni2018inference}
\textsc{Bugni, F.~A., I.~A. Canay, and A.~M. Shaikh} (2018): \enquote{Inference
  Under Covariate-Adaptive Randomization,} \emph{Journal of the American
  Statistical Association}, 113, 1784--1796, pMID: 30906087.

\bibitem[\protect\citeauthoryear{Callen, Gulzar, Hasanain, Khan, and
  Rezaee}{Callen et~al.}{2020}]{callen2020data}
\textsc{Callen, M., S.~Gulzar, A.~Hasanain, M.~Y. Khan, and A.~Rezaee} (2020):
  \enquote{Data and policy decisions: Experimental evidence from Pakistan,}
  \emph{Journal of Development Economics}, 146, 102523.

\bibitem[\protect\citeauthoryear{Carter, Laajaj, and Yang}{Carter
  et~al.}{2021}]{carter2021subsidies}
\textsc{Carter, M., R.~Laajaj, and D.~Yang} (2021): \enquote{Subsidies and the
  African Green Revolution: direct effects and social network spillovers of
  randomized input subsidies in Mozambique,} \emph{American Economic Journal:
  Applied Economics}, 13, 206--229.

\bibitem[\protect\citeauthoryear{Casaburi and Reed}{Casaburi and
  Reed}{2022}]{casaburi2021using}
\textsc{Casaburi, L. and T.~Reed} (2022): \enquote{Using individual-level
  randomized treatment to learn about market structure,} \emph{American
  Economic Journal: Applied Economics}, 14, 58--90.

\bibitem[\protect\citeauthoryear{Crépon, Devoto, Duflo, and Parienté}{Crépon
  et~al.}{2015}]{crepon2015estimating}
\textsc{Crépon, B., F.~Devoto, E.~Duflo, and W.~Parienté} (2015):
  \enquote{Estimating the {Impact} of {Microcredit} on {Those} {Who} {Take}
  {It} {Up}: {Evidence} from a {Randomized} {Experiment} in {Morocco},}
  \emph{American Economic Journal: Applied Economics}, 7, 123--150.

\bibitem[\protect\citeauthoryear{de~Chaisemartin and
  Ramirez-Cuellar}{de~Chaisemartin and Ramirez-Cuellar}{2020}]{de2020level}
\textsc{de~Chaisemartin, C. and J.~Ramirez-Cuellar} (2020): \enquote{At what
  level should one cluster standard errors in paired experiments, and in
  stratified experiments with small strata?} Tech. rep., National Bureau of
  Economic Research.

\bibitem[\protect\citeauthoryear{de~Mel, McKenzie, and Woodruff}{de~Mel
  et~al.}{2019}]{de_mel2019labor}
\textsc{de~Mel, S., D.~McKenzie, and C.~Woodruff} (2019): \enquote{Labor
  {Drops}: {Experimental} {Evidence} on the {Return} to {Additional} {Labor} in
  {Microenterprises},} \emph{American Economic Journal: Applied Economics}, 11,
  202--235.

\bibitem[\protect\citeauthoryear{Dhar, Jain, and Jayachandran}{Dhar
  et~al.}{2022}]{dhar2022reshaping}
\textsc{Dhar, D., T.~Jain, and S.~Jayachandran} (2022): \enquote{Reshaping
  adolescents' gender attitudes: Evidence from a school-based experiment in
  India,} \emph{American economic review}, 112, 899--927.

\bibitem[\protect\citeauthoryear{Donner and Klar}{Donner and
  Klar}{2000}]{donner2000design}
\textsc{Donner, A. and N.~Klar} (2000): \enquote{Design and analysis of cluster
  randomization trials in health research,} .

\bibitem[\protect\citeauthoryear{Duflo, Dupas, and Kremer}{Duflo
  et~al.}{2015}]{duflo2015education}
\textsc{Duflo, E., P.~Dupas, and M.~Kremer} (2015): \enquote{Education, HIV,
  and early fertility: Experimental evidence from Kenya,} \emph{American
  Economic Review}, 105, 2757--97.

\bibitem[\protect\citeauthoryear{Fryer}{Fryer}{2018}]{fryer2018pupil}
\textsc{Fryer, R.} (2018): \enquote{The "{Pupil}" {Factory}: {Specialization}
  and the {Production} of {Human} {Capital} in {Schools},} \emph{American
  Economic Review}, 108, 616--656.

\bibitem[\protect\citeauthoryear{Fukumoto}{Fukumoto}{2022}]{fukumoto2022nonignorable}
\textsc{Fukumoto, K.} (2022): \enquote{Nonignorable Attrition in Pairwise
  Randomized Experiments,} \emph{Political Analysis}, 30, 132--141.

\bibitem[\protect\citeauthoryear{Gerber and Green}{Gerber and
  Green}{2012}]{gerber2012field}
\textsc{Gerber, A.~S. and D.~P. Green} (2012): \emph{Field experiments: Design,
  analysis, and interpretation}, WW Norton.

\bibitem[\protect\citeauthoryear{Ghanem, Hirshleifer, and Ortiz-Becerra}{Ghanem
  et~al.}{2021}]{ghanem2021testing}
\textsc{Ghanem, D., S.~Hirshleifer, and K.~Ortiz-Becerra} (2021):
  \enquote{Testing attrition bias in field experiments,} .

\bibitem[\protect\citeauthoryear{Glennerster and Takavarasha}{Glennerster and
  Takavarasha}{2013}]{glennerster2013running}
\textsc{Glennerster, R. and K.~Takavarasha} (2013): \emph{Running {Randomized}
  {Evaluations}: {A} {Practical} {Guide}}, Princeton University Press.

\bibitem[\protect\citeauthoryear{Groh and McKenzie}{Groh and
  McKenzie}{2016}]{groh2016macroinsurance}
\textsc{Groh, M. and D.~McKenzie} (2016): \enquote{Macroinsurance for
  microenterprises: A randomized experiment in post-revolution Egypt,}
  \emph{Journal of Development Economics}, 118, 13--25.

\bibitem[\protect\citeauthoryear{Heckman}{Heckman}{1979}]{heckman1979sample}
\textsc{Heckman, J.~J.} (1979): \enquote{Sample selection bias as a
  specification error,} \emph{Econometrica: Journal of the econometric
  society}, 153--161.

\bibitem[\protect\citeauthoryear{Hirano, Imbens, Ridder, and Rubin}{Hirano
  et~al.}{2001}]{hirano2001combining}
\textsc{Hirano, K., G.~W. Imbens, G.~Ridder, and D.~B. Rubin} (2001):
  \enquote{Combining Panel Data Sets with Attrition and Refreshment Samples,}
  \emph{Econometrica}, 69, 1645--1659.

\bibitem[\protect\citeauthoryear{Hjort, Moreira, Rao, and Santini}{Hjort
  et~al.}{2021}]{hjort2021research}
\textsc{Hjort, J., D.~Moreira, G.~Rao, and J.~F. Santini} (2021): \enquote{How
  research affects policy: Experimental evidence from 2,150 brazilian
  municipalities,} \emph{American Economic Review}, 111, 1442--1480.

\bibitem[\protect\citeauthoryear{Horowitz and Manski}{Horowitz and
  Manski}{2000}]{manski2000analysis}
\textsc{Horowitz, J.~L. and C.~F. Manski} (2000): \enquote{Nonparametric
  Analysis of Randomized Experiments with Missing Covariate and Outcome Data,}
  \emph{Journal of the American Statistical Association}, 95, 77--84.

\bibitem[\protect\citeauthoryear{King, Gakidou, Ravishankar, Moore, Lakin,
  Vargas, T{\'e}llez-Rojo, Hern{\'a}ndez~{\'A}vila, {\'A}vila, and Llamas}{King
  et~al.}{2007}]{king2007politically}
\textsc{King, G., E.~Gakidou, N.~Ravishankar, R.~T. Moore, J.~Lakin, M.~Vargas,
  M.~M. T{\'e}llez-Rojo, J.~E. Hern{\'a}ndez~{\'A}vila, M.~H. {\'A}vila, and
  H.~H. Llamas} (2007): \enquote{A “politically robust” experimental design
  for public policy evaluation, with application to the Mexican universal
  health insurance program,} \emph{Journal of Policy Analysis and Management},
  26, 479--506.

\bibitem[\protect\citeauthoryear{Lee}{Lee}{2009}]{lee2009training}
\textsc{Lee, D.~S.} (2009): \enquote{Training, {Wages}, and {Sample}
  {Selection}: {Estimating} {Sharp} {Bounds} on {Treatment} {Effects},}
  \emph{The Review of Economic Studies}, 76, 1071--1102.

\bibitem[\protect\citeauthoryear{Little and Rubin}{Little and
  Rubin}{2019}]{little2019statistical}
\textsc{Little, R.~J. and D.~B. Rubin} (2019): \emph{Statistical analysis with
  missing data}, vol. 793, John Wiley \& Sons.

\bibitem[\protect\citeauthoryear{Moffit, Fitzgerald, and Gottschalk}{Moffit
  et~al.}{1999}]{moffit1999sample}
\textsc{Moffit, R., J.~Fitzgerald, and P.~Gottschalk} (1999): \enquote{Sample
  attrition in panel data: The role of selection on observables,} \emph{Annales
  d'Economie et de Statistique}, 129--152.

\bibitem[\protect\citeauthoryear{Romero, Sandefur, and Sandholtz}{Romero
  et~al.}{2020}]{romero2020outsourcing}
\textsc{Romero, M., J.~Sandefur, and W.~A. Sandholtz} (2020):
  \enquote{Outsourcing education: Experimental evidence from Liberia,}
  \emph{American Economic Review}, 110, 364--400.

\bibitem[\protect\citeauthoryear{Rubin}{Rubin}{2004}]{rubin2004multiple}
\textsc{Rubin, D.~B.} (2004): \emph{Multiple imputation for nonresponse in
  surveys}, vol.~81, John Wiley \& Sons.

\bibitem[\protect\citeauthoryear{van~der Vaart}{van~der
  Vaart}{1998}]{van_der_vaart1998asymptotic}
\textsc{van~der Vaart, A.~W.} (1998): \emph{Asymptotic statistics}, vol.~3 of
  \emph{Cambridge {Series} in {Statistical} and {Probabilistic} {Mathematics}},
  Cambridge University Press, Cambridge.

\bibitem[\protect\citeauthoryear{Wooldridge}{Wooldridge}{2002}]{wooldridge2002inverse}
\textsc{Wooldridge, J.~M.} (2002): \enquote{Inverse probability weighted
  M-estimators for sample selection, attrition, and stratification,}
  \emph{Portuguese economic journal}, 1, 117--139.

\end{thebibliography}

\pagebreak

\appendix
\section{Appendix}
\subsection{Proof of Theorem \ref{thm:pair}}
First we argue that $\hat{\theta}_n \stackrel{P}{\to} \theta^{\rm obs}$. Note that
\[\hat{\theta}_n = \frac{\sum_{1 \leq i \leq n} Y_i(1) R_i(1) D_i}{\sum_{1 \leq i \leq n} R_i(1) D_i} - \frac{\sum_{1 \leq i \leq n} Y_i(0) R_i(0) (1 - D_i)}{\sum_{1 \leq i \leq n} R_i(0) (1 - D_i)}~.\]
By Lemma S.1.5 in \cite{bai2021inference},
\begin{align*}
&\frac{1}{n/2}\sum_{1 \le i \le n}Y_i(1) R_i(1) D_i \stackrel{P}{\to}  E[Y_i(1)R_i(1)]~,\\
    &\frac{1}{n/2}\sum_{1 \leq i \leq n}R_i(1)D_i \stackrel{P}{\to}  E[R_i(1)]~, \\
    &\frac{1}{n/2}\sum_{1 \leq i \leq n} Y_i(0) R_i(0) (1 - D_i)  \stackrel{P}{\to} E[Y_i(0)R_i(0)]~, \\
    &\frac{1}{n/2}\sum_{1 \leq i \leq n}R_i(0)D_i \stackrel{P}{\to} E[R_i(0)]~.
\end{align*}
Hence the result follows by the continuous mapping theorem. Next we argue that $\hat{\theta}_n^{\rm drop} \stackrel{P}{\to} \theta^{\rm drop}$. To begin, recall that $\hat \theta_n^{\rm drop} = \mathbb B_n^{-1} \mathbb C_n$, where
\begin{align*}
\mathbb B_n & = \frac{1}{n/2} \sum_{1 \leq j \leq n/2} R_{\pi(2j - 1)} R_{\pi(2j)} \\
\mathbb C_n & = \frac{1}{n/2}\sum_{1 \leq j \leq n/2} R_{\pi(2j - 1)} R_{\pi(2j)} (Y_{\pi(2j - 1)} - Y_{\pi(2j)}) (D_{\pi(2j - 1)} - D_{\pi(2j)})
\end{align*}
For $\mathbb B_n$, it follows Assumptions \ref{as:mp}, \ref{as:Q-lip}(a), \ref{as:close}, the fact that $R_i(d) \in \{0, 1\}$ for $d \in \{0, 1\}$ and therefore has finite second moments, and similar arguments to those in the proof of Lemma S.1.6 of \cite{bai2021inference} that as $n \to \infty$,
\begin{equation} \label{eq:Bn}
\mathbb B_n \stackrel{P}{\to} E[E[R_i(1) | X_i] E[R_i(0) | X_i]]~.
\end{equation}
Next, we turn to $\mathbb C_n$. Note
\begin{multline*}
\mathbb C_n = \frac{1}{n/2} \sum_{1 \leq j \leq n/2} \Big ( R_{\pi(2j - 1)}(1) R_{\pi(2j)}(0) (Y_{\pi(2j - 1)}(1) - Y_{\pi(2j)}(0)) D_{\pi(2j - 1)} \\
+ R_{\pi(2j - 1)}(0) R_{\pi(2j)}(1) (Y_{\pi(2j)}(1) - Y_{\pi(2j - 1)}(0)) (1 - D_{\pi(2j - 1)}) \Big )~.
\end{multline*}
It follows from Assumption \ref{as:mp} and $Q_n = Q^{n}$ that
\begin{align*}
E[\mathbb C_n | X^{(n)}] & = \frac{1}{n} \sum_{1 \leq j \leq n/2} \Big ( E[Y_{\pi(2j - 1)}(1) R_{\pi(2j - 1)}(1)  | X_{\pi(2j - 1)}] E[R_{\pi(2j)}(0) | X_{\pi(2j)}] \\
& \hspace{5em} - E[Y_{\pi(2j)}(0) R_{\pi(2j)}(0)  | X_{\pi(2j)}] E[R_{\pi(2j - 1)}(1) | X_{\pi(2j - 1)}] \\
& \hspace{5em} + E[Y_{\pi(2j)}(1) R_{\pi(2j)}(1)  | X_{\pi(2j)}] E[R_{\pi(2j - 1)}(0) | X_{\pi(2j - 1)}] \\
& \hspace{5em} - E[Y_{\pi(2j - 1)}(0) R_{\pi(2j - 1)}(0)  | X_{\pi(2j - 1)}] E[R_{\pi(2j)}(1) | X_{\pi(2j)}]\Big )
\end{align*}
Next, it follows from Assumptions \ref{as:Q}(a), \ref{as:Q-lip}(b), \ref{as:mp}, \ref{as:close}, and similar arguments to those in the proof of Lemma S.1.6 of \cite{bai2021inference} that as $n \to \infty$,
\begin{multline} \label{eq:Cn1}
\frac{1}{n} \sum_{1 \leq j \leq n/2} \Big ( E[Y_{\pi(2j - 1)}(1) R_{\pi(2j - 1)}(1)  | X_{\pi(2j - 1)}] E[R_{\pi(2j)}(0) | X_{\pi(2j)}] \\
+ E[Y_{\pi(2j)}(1) R_{\pi(2j)}(1)  | X_{\pi(2j)}] E[R_{\pi(2j - 1)}(0) | X_{\pi(2j - 1)}] \Big ) \stackrel{P}{\to} E[E[Y_i(1) R_i(1) | X_i] E[R_i(0) | X_i]]
\end{multline}
and
\begin{multline} \label{eq:Cn2}
\frac{1}{n} \sum_{1 \leq j \leq n/2} \Big ( E[Y_{\pi(2j - 1)}(0) R_{\pi(2j - 1)}(0)  | X_{\pi(2j - 1)}] E[R_{\pi(2j)}(1) | X_{\pi(2j)}] \\
+ E[Y_{\pi(2j)}(0) R_{\pi(2j)}(0)  | X_{\pi(2j)}] E[R_{\pi(2j - 1)}(1) | X_{\pi(2j - 1)}] \Big ) \stackrel{P}{\to} E[E[Y_i(0) R_i(0) | X_i] E[R_i(1) | X_i]]~.
\end{multline}
Moreover, it can be shown using similar arguments to those in the proof of Lemma S.1.6 of \cite{bai2021inference} that 
\begin{equation}\label{eq:Cn3}
\left|\mathbb{C}_n - E[\mathbb{C}_n|X^{(n)}]\right| \stackrel{P}{\to} 0~,
\end{equation}
and hence by combining \eqref{eq:Cn1}-\eqref{eq:Cn3} we obtain that
\begin{equation} \label{eq:Cn}
\mathbb{C}_n \stackrel{P}{\to} [E[Y_i(1) R_i(1) | X_i] E[R_i(0) | X_i]] + E[E[Y_i(0) R_i(0) | X_i] E[R_i(1) | X_i]]~.
\end{equation}
The conclusion then follows from \eqref{eq:Bn}, \eqref{eq:Cn}, as well as the continuous mapping theorem.
\qed

\subsection{Proof of Theorem \ref{thm:covariate-adaptive}}
First we argue that $\hat{\theta}_n \stackrel{P}{\to} \theta^{\rm obs}$. Note that
\[\hat{\theta}_n = \frac{\sum_{1 \leq i \leq n} Y_i(1) R_i(1) D_i}{\sum_{1 \leq i \leq n} R_i(1) D_i} - \frac{\sum_{1 \leq i \leq n} Y_i(0) R_i(0) (1 - D_i)}{\sum_{1 \leq i \leq n} R_i(0) (1 - D_i)}~.\]
By Lemma B.3 in \cite{bugni2018inference},
\[\frac{1}{n}\sum_{1 \le i \le n}Y_i(1) R_i(1) D_i \stackrel{P}{\to} \nu E[Y_i(1)R_i(1)]~,\]
where we note that an inspection of their proof shows that Assumption \ref{as:strat}(b) is sufficient to establish their result. Similarly,
\begin{align*}
    &\frac{1}{n}\sum_{1 \leq i \leq n}R_i(1)D_i \stackrel{P}{\to} \nu E[R_i(1)]~, \\
    &\frac{1}{n}\sum_{1 \leq i \leq n} Y_i(0) R_i(0) (1 - D_i)  \stackrel{P}{\to} (1 - \nu) E[Y_i(0)R_i(0)]~, \\
    &\frac{1}{n}\sum_{1 \leq i \leq n}R_i(0)D_i \stackrel{P}{\to} (1 - \nu)E[R_i(0)]~.
\end{align*}
Hence the result follows by the continuous mapping theorem. Next we argue that $\hat{\theta}_n^{\rm sfe} \stackrel{P}{\to} \theta^{\rm sfe}$. To that end, write $\hat{\theta}_n^{\rm sfe}$ as 
\begin{equation*}
    \hat{\theta}_n^{\rm sfe} = \frac{\sum_{1\leq i\leq n} R_i \Tilde{D}_i Y_i}{\sum_{1\leq i\leq n} R_i \Tilde{D}_i^2}~,
\end{equation*}
where $\Tilde{D}_i$ is the projection of $D_i$ on the strata indicators, i.e., $\Tilde{D}_i = D_i - n_1(S_i)/n(S_i)$, and
\begin{equation*}
    \frac{n_1(S_i)}{n(S_i)} = \sum_{s \in \mathcal{S}}  I \{S_i = s\} \frac{n_1(s)}{n(s)}~,
\end{equation*}
for
\begin{equation*}
    n_1(s) = \sum_{1\leq i \leq n }   R_i D_i I \{S_i = s\}, \quad n(s) = \sum_{1\leq i \leq n } R_i  I \{S_i = s\} ~.
\end{equation*}
By Lemma B.3 in \cite{bugni2018inference} and the continuous mapping theorem, we have 
\begin{align*}
    \frac{n_1(s)}{n(s)} = \frac{\frac{1}{n}\sum_{1\leq i \leq n }  R_i D_i I \{S_i = s\}}{\frac{1}{n}\sum_{1\leq i \leq n }  R_i I \{S_i = s\}} & \xrightarrow{P}  \frac{\nu E[R_i(1)   I \{S_i = s\}]}{\nu E[R_i(1) I \{S_i = s\}] + (1 - \nu) E[R_i(0) I \{S_i = s\}]} \\
    & = \frac{\nu E[R_i(1) | S_i = s]}{\nu E[R_i(1) | S_i = s] + (1 - \nu) E[R_i(0) | S_i = s]}~.
\end{align*}
Similarly,
\begin{align*}
    &\frac{1}{n} \sum_{1\leq i \leq n} R_i \Tilde{D}_i Y_i \\
    &= \frac{1}{n} \sum_{1\leq i \leq n} R_i D_i Y_i - \sum_{s \in \mathcal S} \frac{1}{n} \sum_{1\leq i \leq n} I \{S_i = s\}  R_i  Y_i \frac{\nu E[R_i(1) | S_i = s]}{\nu E[R_i(1) | S_i = s] + (1 - \nu) E[R_i(0) | S_i = s]} \\
    & \hspace{3em} +  \sum_{s \in \mathcal{S}} \frac{1}{n} \sum_{1\leq i \leq n} R_i Y_i I\{S_i=s\}  \left(\frac{\nu E[R_i(1) | S_i = s]}{\nu E[R_i(1) | S_i = s] + (1 - \nu) E[R_i(0) | S_i = s]}-\frac{n_1(s)}{n(s)} \right)\\
    &\xrightarrow{P} \nu E[R_i(1)Y_i(1)] - \sum_{s \in \mathcal S} (\nu E[R_i(1) Y_i(1) I \{S_i = s\}] + (1 - \nu) E[R_i(0) Y_i(0) I \{S_i = s\}]) \\
    & \hspace{3em} \times \frac{\nu E[R_i(1) | S_i = s]}{\nu E[R_i(1) | S_i = s] + (1 - \nu) E[R_i(0) | S_i = s]} \\
    & = \nu E[R_i(1)Y_i(1)] - \sum_{s \in \mathcal S} p(s) (\nu E[R_i(1) Y_i(1) | S_i = s] + (1 - \nu) E[R_i(0) Y_i(0) | S_i = s]) \\
    & \hspace{3em} \times \frac{\nu E[R_i(1) | S_i = s]}{\nu E[R_i(1) | S_i = s] + (1 - \nu) E[R_i(0) | S_i = s]} \\
    & = \nu E[R_i(1)Y_i(1)] \\
    & \hspace{3em} - E \left [ (\nu E[R_i(1) Y_i(1) | S_i] + (1 - \nu) E[R_i(0) Y_i(0) | S_i]) \frac{\nu E[R_i(1) | S_i]}{\nu E[R_i(1) | S_i] + (1 - \nu) E[R_i(0) | S_i]} \right ] \\
    & = \nu (1 - \nu) E \left [ \frac{E[R_i(1) Y_i(1) | S_i] E[R_i(0) | S_i] - E[R_i(0) Y_i(0) | S_i] E[R_i(1) | S_i]}{\nu E[R_i(1) | S_i] + (1 - \nu) E[R_i(0) | S_i]} \right ]~,
\end{align*}
where in the last equality we used the fact that $E[R_i(1) Y_i(1)] = E[E[R_i(1) Y_i(1) | S_i]]$. Also note that
\begin{align*}
    &\frac{1}{n}\sum_{1\leq i\leq n} R_i \Tilde{D}_i^2\\
    &= \frac{1}{n}\sum_{1\leq i\leq n} R_i \Tilde{D}_i \left(D_i - \frac{n_1(S_i)}{n(S_i)} \right) \\
    &= \frac{1}{n}\sum_{1\leq i\leq n} R_i \left(1 - \frac{n_1(S_i)}{n(S_i)} \right) D_i  \\
    &= \frac{1}{n}\sum_{1\leq i\leq n} R_i D_i  - \sum_{s \in \mathcal{S}} \frac{n_1(s)}{n(s)} \frac{1}{n}\sum_{1\leq i\leq n} R_i D_i   I \{S_i = s\} \\
    & \stackrel{P}{\to} \nu E[R_i(1)] - \sum_{s \in \mathcal S} p(s) \nu E[R_i(1) | S_i = s] \frac{\nu E[R_i(1) | S_i = s]}{\nu E[R_i(1) | S_i = s] + (1 - \nu) E[R_i(0) | S_i = s]} \\
    & = \nu E[R_i(1)] - \nu E \left [ E[R_i(1) | S_i] \frac{\nu E[R_i(1) | S_i]}{\nu E[R_i(1) | S_i] + (1 - \nu) E[R_i(0) | S_i]} \right ] \\
    & = \nu (1 - \nu) E \left [ \frac{E[R_i(1) | S_i] E[R_i(0) | S_i]}{\nu E[R_i(1) | S_i] + (1 - \nu) E[R_i(0) | S_i]} \right ]~,
\end{align*}
where the second equality follows from $\sum_{1\leq i\leq n} R_i \Tilde{D}_i\frac{n_1(S_i)}{n(S_i)} = 0$, which is derived as follows:
\begin{align*}
     &\sum_{1\leq i\leq n} R_i \Tilde{D}_i\frac{n_1(S_i)}{n(S_i)} = \sum_{1\leq i\leq n} R_i \Tilde{D}_i \sum_{s \in \mathcal{S}}  I \{S_i = s\} \frac{n_1(s)}{n(s)} = \sum_{s \in \mathcal{S}}  \frac{n_1(s)}{n(s)}\sum_{1\leq i\leq n} R_i \Tilde{D}_i   I \{S_i = s\} \\
     &= \sum_{s \in \mathcal{S}}  \frac{n_1(s)}{n(s)}\sum_{1\leq i\leq n} R_i D_i   I \{S_i = s\} - \sum_{s \in \mathcal{S}}  \frac{n_1(s)}{n(s)}\sum_{1\leq i\leq n} R_i I \{S_i = s\}  \frac{n_1(S_i)}{n(S_i)} \\
     &= \sum_{s \in \mathcal{S}}  \frac{n_1(s)}{n(s)} n_1(s) - \sum_{s \in \mathcal{S}}  \frac{n_1(s)}{n(s)}\sum_{1\leq i\leq n} R_i I \{S_i = s\}  \sum_{k \in \mathcal{S}}  I \{S_i = k\} \frac{n_1(k)}{n(k)} \\
     &= \sum_{s \in \mathcal{S}} \frac{n_1(s)^2}{n(s)} - \sum_{s \in \mathcal{S}}  \frac{n_1(s)}{n(s)}\sum_{1\leq i\leq n} R_i I \{S_i = s\} \frac{n_1(s)}{n(s)} \\
     &= \sum_{s \in \mathcal{S}}  \frac{n_1(s)^2}{n(s)} - \sum_{s \in \mathcal{S}}  \frac{n_1(s)}{n(s)} n(s) \frac{n_1(s)}{n(s)} = 0~.
\end{align*}
The conclusion then follows from the continuous mapping theorem. \qed

\subsection{The Limiting Distribution of $\hat \theta_n$}\label{sec:norm_theta}
\begin{theorem} \label{thm:pair-distr}
Suppose $Q$ satisfies Assumption \ref{as:Q} (as well as $E[Y_i^2(d)] < \infty$) and Assumption \ref{as:Q-lip} (as well as $E[Y_i^2(d)R_i(d)|X_i = x]$ is Lipschitz for $d \in \{0, 1\}$), and the treatment assignment mechanism satisfies Assumptions \ref{as:mp}, \ref{as:close} as well as
\[ \frac{1}{n} \sum_{1 \leq j \leq n} ||X_{\pi(2j -1)} - X_{\pi(2j)}||^2 \stackrel{P}{\to} 0~. \]
Then, as $n \to \infty$,
\[ \sqrt n(\hat \theta_n - \theta(Q)) \stackrel{d}{\to} N(0, \varsigma_{\rm mp}^2)~, \]
where
\[ \varsigma_{\rm mp}^2 = \var[\tilde Y_i(1)] + \var[\tilde Y_i(0)] - \frac{1}{2} E[E[\tilde Y_i(1) + \tilde Y_i(0) | X_i]^2] \]
and
\[ \tilde Y_i(d) = \frac{R_i(d)}{E[R_i(d)]} \left ( Y_i(d) - \frac{E[Y_i(d) R_i(d)]}{E[R_i(d)]} \right ) \]
for $d \in \{0, 1\}$.
\end{theorem}

\begin{remark}
Following arguments similar to those in \cite{bai2023inference}, we can construct a consistent estimator of $\varsigma^2_{\rm mp}$. To that end, consider the observed adjusted outcome defined as:
\begin{align*}
    \hat Y_{i} = \frac{R_i}{ \frac{1}{n}\sum_{1\leq j \leq 2n} R_j I\{D_j = D_i\}}\left(Y_{i}-\frac{\frac{1}{n}\sum_{1\leq j \leq 2n} Y_{j} I\{D_j = D_i\} R_j}{\frac{1}{n}\sum_{1\leq j \leq 2n} I\{D_j = D_i\} R_j}\right)~,
\end{align*}
We then propose the following variance estimator:
\begin{equation}\label{eqn:define-variance-estimator}
    \hat{v}_{n}^{2}=\hat{\tau}_{n}^{2}-\frac{1}{2}\hat{\lambda}_{n}^{2}~,
\end{equation}
where
\begin{align*}
\hat{\tau}_{n}^{2}=& \frac{1}{n} \sum_{1 \leq j \leq n}\left(\hat Y_{\pi(2 j)}-\hat Y_{\pi(2 j-1)}\right)^{2} \\
\hat{\lambda}_{n}^{2}=& \frac{2}{n} \sum_{1 \leq j \leq\lfloor n/2\rfloor} \left(\hat Y_{\pi(4 j-3)}-\hat Y_{\pi(4 j-2)}\right) \left(\hat Y_{\pi(4 j-1)}- \hat Y_{\pi(4j)} \right)(D_{\pi(4 j-3)}-D_{\pi(4 j-2)})(D_{\pi(4 j-1)}-D_{\pi(4 j)})~.
\end{align*}
It follows from similar arguments to those used in \cite{bai2023inference} that under appropriate assumptions $\hat v_n^2 \stackrel{P}{\to} \varsigma_{\rm mp}^2$.
\end{remark}

\begin{proof}[\sc Proof of Theorem \ref{thm:pair-distr}]
To begin, note
\[ \hat \theta_n = \frac{\frac{1}{n}\sum_{1 \leq i \leq 2n} Y_i(1) R_i(1) D_i}{\frac{1}{n}\sum_{1 \leq i \leq 2n} R_i(1) D_i} - \frac{\frac{1}{n}\sum_{1 \leq i \leq 2n} Y_i(0) R_i(0) (1 - D_i)}{\frac{1}{n}\sum_{1 \leq i \leq 2n} R_i(0) (1 - D_i)}~. \]
Next, note by Assumption \ref{as:mp} that
\[ \sqrt n \Big ( \frac{1}{n}\sum_{1 \leq i \leq 2n} Y_i(1) R_i(1) D_i - E[Y_i(1) R_i(1)] \Big ) = \frac{1}{\sqrt n} \sum_{1 \leq i \leq 2n} (Y_i(1) R_i(1) D_i - E[Y_i(1) R_i(1)] D_i) \]
and similarly for the other three terms. The desired conclusion then follows from Lemma \ref{lem:L} together with an application of the delta method. In particular, for $g(x, y, z, w) = \frac{x}{y} - \frac{z}{w}$, observe that
\begin{align*}
\hat \theta_n &= g \left( \frac{1}{n}\sum_{1 \leq i \leq 2n} Y_i(1) R_i(1) D_i, \frac{1}{n}\sum_{1 \leq i \leq 2n} R_i(1) D_i, \frac{1}{n}\sum_{1 \leq i \leq 2n} Y_i(0) R_i(0) (1 - D_i), \frac{1}{n}\sum_{1 \leq i \leq 2n} R_i(0) (1 - D_i) \right)
\end{align*}
and the Jacobian is
\[ D_g(x, y, z, w) = \Big ( \frac{1}{y}, - \frac{x}{y^2}, - \frac{1}{w}, \frac{z}{w^2} \Big )~. \]
Note by the laws of total variance and total covariance that $\mathbb V$ in Lemma \ref{lem:L} is symmetric with entries
\begin{align*}
\mathbb V_{11} & = \var[Y_i(1) R_i(1)] - \frac{1}{2} \var[E[Y_i(1) R_i(1) | X_i]] \\
\mathbb V_{12} & = \cov[Y_i(1) R_i(1), R_i(1)] - \frac{1}{2} \cov[E[Y_i(1) R_i(1) | X_i], E[R_i(1) | X_i]] \\
\mathbb V_{13} & = \frac{1}{2} \cov[E[Y_i(1) R_i(1) | X_i], E[Y_i(0) R_i(0) | X_i]] \\
\mathbb V_{14} & = \frac{1}{2} \cov[E[Y_i(1) R_i(1) | X_i], E[R_i(0) | X_i]] \\
\mathbb V_{22} & = \var[R_i(1)] - \frac{1}{2} \var[E[R_i(1) | X_i]] \\
\mathbb V_{23} & = \frac{1}{2} \cov[E[R_i(1) | X_i], E[Y_i(0) R_i(0) | X_i]] \\
\mathbb V_{24} & = \frac{1}{2} \cov[E[R_i(1) | X_i], E[R_i(0) | X_i]] \\
\mathbb V_{33} & = \var[Y_i(0) R_i(0)] - \frac{1}{2} \var[E[Y_i(0) R_i(0) | X_i]] \\
\mathbb V_{34} & = \cov[Y_i(0) R_i(0), R_i(0)] - \frac{1}{2} \cov[E[Y_i(0) R_i(0) | X_i], E[R_i(0) | X_i]] \\
\mathbb V_{44} & = \var[R_i(0)] - \frac{1}{2} \var[E[R_i(0) | X_i]]~.
\end{align*}
The conclusion of the theorem then follows from direct calculation.
\end{proof}

\begin{lemma} \label{lem:L}
Suppose $Q$ satisfies Assumption \ref{as:Q} (as well as $E[Y_i^2(d)] < \infty$) and Assumption \ref{as:Q-lip} (as well as $E[Y_i^2(d)R_i(d)|X_i = x]$ is Lipschitz for $d \in \{0, 1\}$), and the treatment assignment mechanism satisfies Assumptions \ref{as:mp}, \ref{as:close} as well as
\begin{equation}\label{eq:close_2}
\frac{1}{n} \sum_{1 \leq j \leq n} \|X_{\pi(2j -1)} - X_{\pi(2j)}\|^2 \stackrel{P}{\to} 0~. 
\end{equation}
Define
\begin{align*}
\mathbb L_n^{\rm YA1} & = \frac{1}{\sqrt n} \sum_{1 \leq i \leq 2n} (Y_i(1) R_i(1) D_i - E[Y_i(1) R_i(1)] D_i) \\
\mathbb L_n^{\rm A1} & = \frac{1}{\sqrt n} \sum_{1 \leq i \leq 2n} (R_i(1) D_i - E[R_i(1)] D_i) \\
\mathbb L_n^{\rm YA0} & = \frac{1}{\sqrt n} \sum_{1 \leq i \leq 2n} (Y_i(0) R_i(0) (1 - D_i) - E[Y_i(0) R_i(0)] (1 - D_i)) \\
\mathbb L_n^{\rm A0} & = \frac{1}{\sqrt n} \sum_{1 \leq i \leq 2n} (R_i(0) (1 - D_i) - E[R_i(0)] (1 - D_i))~.
\end{align*}
Then, as $n \to \infty$,
\[ (\mathbb L_n^{\rm YA1}, \mathbb L_n^{\rm A1}, \mathbb L_n^{\rm YA0}, \mathbb L_n^{\rm A0})' \stackrel{d}{\to} N(0, \mathbb V)~, \] where
\[ \mathbb V = \mathbb V_1 + \mathbb V_2 \]
for
\[ \mathbb V_1 = \begin{pmatrix}
\mathbb V_1^1 & 0 \\
0 & \mathbb V_1^0
\end{pmatrix} \]
\begin{align*}
\mathbb V_1^1 & = \begin{pmatrix}
E[\var[Y_i(1) R_i(1) | X_i]] & E[\cov[Y_i(1) R_i(1), R_i(1) | X_i]] \\
E[\cov[Y_i(1) R_i(1), R_i(1) | X_i]] & E[\var[R_i(1) | X_i]]
\end{pmatrix} \\
\mathbb V_1^0 & = \begin{pmatrix}
E[\var[Y_i(0) R_i(0) | X_i]] & E[\cov[Y_i(0) R_i(0), R_i(0) | X_i]] \\
E[\cov[Y_i(0) R_i(0), R_i(0) | X_i]] & E[\var[R_i(0) | X_i]]
\end{pmatrix}
\end{align*}
\[ \mathbb V_2 = \frac{1}{2} \var[(E[Y_i(1) R_i(1) | X_i], E[R_i(1) | X_i], E[Y_i(0) R_i(0) | X_i], E[R_i(0) | X_i])']~. \]
\end{lemma}

\begin{proof}[\sc Proof of Lemma \ref{lem:L}]
Note
\[ (\mathbb L_n^{\rm YA1}, \mathbb L_n^{\rm A1}, \mathbb L_n^{\rm YA0}, \mathbb L_n^{\rm A0}) = (\mathbb L_{1, n}^{\rm YA1}, \mathbb L_{1, n}^{\rm A1}, \mathbb L_{1, n}^{\rm YA0}, \mathbb L_{1, n}^{\rm A0}) + (\mathbb L_{2, n}^{\rm YA1}, \mathbb L_{2, n}^{\rm A1}, \mathbb L_{2, n}^{\rm YA0}, \mathbb L_{2, n}^{\rm A0})~, \]
where
\begin{align*}
\mathbb L_{1, n}^{\rm YA1} & = \frac{1}{\sqrt n} \sum_{1 \leq i \leq 2n} (Y_i(1) R_i(1) D_i - E[Y_i(1) R_i(1) D_i | X^{(n)}, D^{(n)}]) \\
\mathbb L_{2, n}^{\rm YA1} & = \frac{1}{\sqrt n} \sum_{1 \leq i \leq 2n} (E[Y_i(1) R_i(1) D_i | X^{(n)}, D^{(n)}] - E[Y_i(1) R_i(1)] D_i)
\end{align*}
and similarly for the rest. Next, note $(\mathbb L_{1, n}^{\rm YA1}, \mathbb L_{1, n}^{\rm A1}, \mathbb L_{1, n}^{\rm YA0}, \mathbb L_{1, n}^{\rm A0}), n \geq 1$ is a triangular array of normalized sums of random vectors. We will apply the Lindeberg central limit theorem for random vectors, i.e., Proposition 2.27 of \cite{van_der_vaart1998asymptotic}, to this triangular array. Conditional on $X^{(n)}, D^{(n)}$, $(\mathbb L_{1, n}^{\rm YA1}, \mathbb L_{1, n}^{\rm A1}) \indep (\mathbb L_{1, n}^{\rm YA0}, \mathbb L_{1, n}^{\rm A0})$. Moreover, it follows from $Q_n = Q^{2n}$ and Assumption \ref{as:mp} that
\begin{multline*}
\var \left [ \begin{pmatrix}
\mathbb L_{1, n}^{\rm YA1} \\
\mathbb L_{1, n}^{\rm A1}
\end{pmatrix} \Bigg |  X^{(n)}, D^{(n)} \right ] \\
= \begin{pmatrix}
\frac{1}{n} \sum_{1 \leq i \leq 2n} \var[Y_i(1) R_i(1) | X_i] D_i & \frac{1}{n} \sum_{1 \leq i \leq 2n} \cov[Y_i(1) R_i(1), R_i(1) | X_i] D_i \\
\frac{1}{n} \sum_{1 \leq i \leq 2n} \cov[Y_i(1) R_i(1), R_i(1) | X_i] D_i & \frac{1}{n} \sum_{1 \leq i \leq 2n} \var[R_i(1) | X_i] D_i
\end{pmatrix}~.
\end{multline*}
For the upper left component, we have
\begin{equation} \label{eq:condlip}
\frac{1}{n} \sum_{1 \leq i \leq 2n} \var[Y_i(1) R_i(1) | X_i] D_i = \frac{1}{n} \sum_{1 \leq i \leq 2n} E[Y_i^2(1) R_i(1) | X_i] D_i - \frac{1}{n} \sum_{1 \leq i \leq 2n} E[Y_i(1) R_i(1) | X_i]^2 D_i~.
\end{equation}
Note
\begin{align*}
& \frac{1}{n} \sum_{1 \leq i \leq 2n} E[Y_i^2(1) R_i(1) | X_i] D_i \\
& = \frac{1}{2n} \sum_{1 \leq i \leq 2n} E[Y_i^2(1) R_i(1) | X_i] + \frac{1}{2} \Big ( \frac{1}{n} \sum_{1 \leq i \leq 2n: D_i = 1} E[Y_i^2(1) R_i(1) | X_i] \\
& \hspace{8cm} - \frac{1}{n} \sum_{1 \leq i \leq 2n: D_i = 0} E[Y_i^2(1) R_i(1) | X_i] \Big )~.
\end{align*}
It follows from the weak law of large numbers, the application of which is permitted by $E[Y_i^2(1)] < \infty$ and the fact that $R_i(1) \in \{0, 1\}$, that
\[ \frac{1}{2n} \sum_{1 \leq i \leq 2n} E[Y_i^2(1) R_i(1) | X_i] \stackrel{P}{\to} E[Y_i^2(1) R_i(1)]~. \]
On the other hand, it follows from Assumption \ref{as:Q-lip} and \ref{as:close} that
\begin{align*}
& \Big | \frac{1}{n} \sum_{1 \leq i \leq 2n: D_i = 1} E[Y_i^2(1) R_i(1) | X_i] - \frac{1}{n} \sum_{1 \leq i \leq 2n: D_i = 0} E[Y_i^2(1) R_i(1) | X_i] \Big | \\
& \leq \frac{1}{n} \sum_{1 \leq j \leq n} |E[Y_{\pi(2j - 1)}^2(1) A_{\pi(2j - 1)}(1) | X_{\pi(2j - 1)}] - E[Y_{\pi(2j)}^2(1) A_{\pi(2j)}(1) | X_{\pi(2j)}]| \\
& \lesssim \frac{1}{n} \sum_{1 \leq j \leq n} \|X_{\pi(2j - 1)} - X_{\pi(2j)}\| = o_P(1)~.
\end{align*}
Therefore,
\[ \frac{1}{n} \sum_{1 \leq i \leq 2n} E[Y_i^2(1) R_i(1) | X_i] D_i \stackrel{P}{\to} E[Y_i^2(1) R_i(1)]~. \]
Meanwhile,
\begin{align*}
& \frac{1}{n} \sum_{1 \leq i \leq 2n} E[Y_i(1) R_i(1) | X_i]^2 D_i \\
& = \frac{1}{2n} \sum_{1 \leq i \leq 2n} E[Y_i(1) R_i(1) | X_i]^2 + \frac{1}{2} \Big ( \frac{1}{n} \sum_{1 \leq i \leq 2n: D_i = 1} E[Y_i(1) R_i(1) | X_i]^2 \\
& \hspace{8cm} - \frac{1}{n} \sum_{1 \leq i \leq 2n: D_i = 0} E[Y_i(1) R_i(1) | X_i]^2 \Big )~.
\end{align*}
Jensen's inequality implies $E[E[Y_i(1) R_i(1) | X_i]^2] \leq E[Y_i^2(1) R_i(1)] < E[Y_i^2(1)] < \infty$, so it follows from the weak law of large numbers as above that
\[ \frac{1}{2n} \sum_{1 \leq i \leq 2n} E[Y_i(1) R_i(1) | X_i]^2 \stackrel{P}{\to} E[E[Y_i(1) R_i(1) | X_i]^2]~. \]
Next,
\begin{align*}
& \Big | \frac{1}{n} \sum_{1 \leq i \leq 2n: D_i = 1} E[Y_i(1) R_i(1) | X_i]^2 - \frac{1}{n} \sum_{1 \leq i \leq 2n: D_i = 0} E[Y_i(1) R_i(1) | X_i]^2 \Big | \\
& \leq \frac{1}{n} \sum_{1 \leq j \leq n} |E[Y_{\pi(2j - 1)}(1) A_{\pi(2j - 1)}(1) | X_{\pi(2j - 1)}] - E[Y_{\pi(2j)}(1) A_{\pi(2j)}(1) | X_{\pi(2j)}]| \\
& \hspace{5em} \times |E[Y_{\pi(2j - 1)}(1) A_{\pi(2j - 1)}(1) | X_{\pi(2j - 1)}] + E[Y_{\pi(2j)}(1) A_{\pi(2j)}(1) | X_{\pi(2j)}]| \\
& \lesssim \Big ( \frac{1}{n} \sum_{1 \leq j \leq n} \|X_{\pi(2j - 1)} - X_{\pi(2j)}\|^2 \Big )^{1/2} \\
& \hspace{5em} \times \Big (\frac{1}{n} \sum_{1 \leq j \leq n} |E[Y_{\pi(2j - 1)}(1) A_{\pi(2j - 1)}(1) | X_{\pi(2j - 1)}] + E[Y_{\pi(2j)}(1) A_{\pi(2j)}(1) | X_{\pi(2j)}]|^2 \Big )^{1/2} \\
& \lesssim \Big ( \frac{1}{n} \sum_{1 \leq j \leq n} \|X_{\pi(2j - 1)} - X_{\pi(2j)}\|^2 \Big )^{1/2} \\
& \hspace{5em} \times \Big (\frac{1}{n} \sum_{1 \leq j \leq n} (|E[Y_{\pi(2j - 1)}(1) A_{\pi(2j - 1)}(1) | X_{\pi(2j - 1)}]|^2 + |E[Y_{\pi(2j)}(1) A_{\pi(2j)}(1) | X_{\pi(2j)}]|^2) \Big )^{1/2} \\
& \leq \Big ( \frac{1}{n} \sum_{1 \leq j \leq n} \|X_{\pi(2j - 1)} - X_{\pi(2j)}\|^2 \Big )^{1/2} \Big ( \frac{1}{n} \sum_{1 \leq i \leq 2n} E[Y_i(1) R_i(1) | X_i]^2 \Big )^{1/2} = o_P(1)~,
\end{align*}
where the first inequality follows by inspection, the second follows from Assumption \ref{as:Q-lip} and the Cauchy-Schwarz inequality, the third follows from $(a + b)^2 \leq 2 a^2 + 2 b^2$, the last follows by inspection again, and the convergence in probability follows from \eqref{eq:close_2}. Therefore, it follows from \eqref{eq:condlip} that
\[ \frac{1}{n} \sum_{1 \leq i \leq 2n} \var[Y_i(1) R_i(1) | X_i] D_i \stackrel{P}{\to} E[\var[Y_i(1) R_i(1) | X_i]]~. \]
Similar arguments imply that
\[ \var \left [ \begin{pmatrix}
\mathbb L_{1, n}^{\rm YA1} \\
\mathbb L_{1, n}^{\rm A1}
\end{pmatrix} \Bigg |  X^{(n)}, D^{(n)} \right ] \stackrel{P}{\to} \mathbb V_1^1~. \]
Similarly,
\[ \var \left [ \begin{pmatrix}
\mathbb L_{1, n}^{\rm YA0} \\
\mathbb L_{1, n}^{\rm A0}
\end{pmatrix} \Bigg |  X^{(n)}, D^{(n)} \right ] \stackrel{P}{\to} \mathbb V_1^0~. \]
If $E[\var[Y_i(1) R_i(1) | X_i]] = E[\var[R_i(1) | X_i]] = E[\var[Y_i(0) R_i(0) | X_i]] = E[\var[R_i(0) | X_i]] = 0$, then it follows from Markov's inequality conditional on $X^{(n)}$ and $D^{(n)}$, and the fact that probabilities are bounded and hence uniformly integrable, that $(\mathbb L_{1, n}^{\rm YA1}, \mathbb L_{1, n}^{\rm A1}, \mathbb L_{1, n}^{\rm YA0}, \mathbb L_{1, n}^{\rm A0}) \stackrel{P}{\to} 0$. Otherwise, it follows from similar arguments to those in the proof of Lemma S.1.5 of \cite{bai2021inference} that
\begin{equation} \label{eq:cond}
\rho(\mathcal L((\mathbb L_{1, n}^{\rm YA1}, \mathbb L_{1, n}^{\rm A1}, \mathbb L_{1, n}^{\rm YA0}, \mathbb L_{1, n}^{\rm A0})' | X^{(n)}, D^{(n)}),  N(0, \mathbb V_1)) \stackrel{P}{\to} 0~,
\end{equation}
where $\mathcal L$ denotes the distribution and $\rho$ is any metric that metrizes weak convergence.

Next, we study $(\mathbb L_{2, n}^{\rm YA1}, \mathbb L_{2, n}^{\rm A1}, \mathbb L_{2, n}^{\rm YA0}, \mathbb L_{2, n}^{\rm A0})$. It follows from $Q_n = Q^{2n}$ and Assumption \ref{as:mp} that
\[ \begin{pmatrix}
\mathbb L_{2, n}^{\rm YA1} \\
\mathbb L_{2, n}^{\rm A1} \\
\mathbb L_{2, n}^{\rm YA0} \\
\mathbb L_{2, n}^{\rm A0}
\end{pmatrix} = \begin{pmatrix}
\frac{1}{\sqrt n} \sum_{1 \leq i \leq 2n} D_i (E[Y_i(1) R_i(1) | X_i] - E[Y_i(1) R_i(1)]) \\
\frac{1}{\sqrt n} \sum_{1 \leq i \leq 2n} D_i (E[R_i(1) | X_i] - E[R_i(1)]) \\
\frac{1}{\sqrt n} \sum_{1 \leq i \leq 2n} (1 - D_i) (E[Y_i(0) R_i(0) | X_i] - E[Y_i(0) R_i(0)]) \\
\frac{1}{\sqrt n} \sum_{1 \leq i \leq 2n} (1 - D_i) (E[R_i(0) | X_i] - E[R_i(0)])
\end{pmatrix}~. \]
For $\mathbb L_{2, n}^{\rm YA1}$, note it follows from Assumptions \ref{as:mp}, \ref{as:Q-lip} and \eqref{eq:close_2} that
\begin{align*}
\var[\mathbb L_{2, n}^{\rm YA1} | X^{(n)}] & = \frac{1}{4n} \sum_{1 \leq j \leq n} (E[Y_{\pi(2j - 1)}(1) A_{\pi(2j - 1)}(1) | X_{\pi(2j - 1)}] - E[Y_{\pi(2j)}(1) A_{\pi(2j)}(1) | X_{\pi(2j)}])^2 \\
& \lesssim \frac{1}{n} \sum_{1 \leq j \leq n} \|X_{\pi(2j - 1)} - X_{\pi(2j)}\|^2 \stackrel{P}{\to} 0~.
\end{align*}
Therefore, it follows from Markov's inequality conditional on $X^{(n)}$ and $D^{(n)}$, and the fact that probabilities are bounded and hence uniformly integrable, that
\[ \mathbb L_{2, n}^{\rm YA1} = E[\mathbb L_{2, n}^{\rm YA1} | X^{(n)}] + o_P(1)~. \]
Similarly,
\[ \begin{pmatrix}
\mathbb L_{2, n}^{\rm YA1} \\
\mathbb L_{2, n}^{\rm A1} \\
\mathbb L_{2, n}^{\rm YA0} \\
\mathbb L_{2, n}^{\rm A0}
\end{pmatrix} = \begin{pmatrix}
\frac{1}{2\sqrt n} \sum_{1 \leq i \leq 2n} (E[Y_i(1) R_i(1) | X_i] - E[Y_i(1) R_i(1)]) \\
\frac{1}{2\sqrt n} \sum_{1 \leq i \leq 2n} (E[R_i(1) | X_i] - E[R_i(1)]) \\
\frac{1}{2\sqrt n} \sum_{1 \leq i \leq 2n} (E[Y_i(0) R_i(0) | X_i] - E[Y_i(0) R_i(0)]) \\
\frac{1}{2\sqrt n} \sum_{1 \leq i \leq 2n} (E[R_i(0) | X_i] - E[R_i(0)])
\end{pmatrix} + o_P(1)~. \]
It then follows from Assumption \ref{as:Q}and the central limit theorem that
\[ (\mathbb L_{2, n}^{\rm YA1}, \mathbb L_{2, n}^{\rm A1}, \mathbb L_{2, n}^{\rm YA0}, \mathbb L_{2, n}^{\rm A0})' \stackrel{d}{\to} N(0, \mathbb V_2)~. \]
Because \eqref{eq:cond} holds and $(\mathbb L_{2, n}^{\rm YA1}, \mathbb L_{2, n}^{\rm A1}, \mathbb L_{2, n}^{\rm YA0}, \mathbb L_{2, n}^{\rm A0})$ is deterministic conditional on $X^{(n)}, D^{(n)}$, the conclusion of the theorem follows from Lemma S.1.3 in \cite{bai2021inference}.
\end{proof}

\subsection{A Numerical Example} \label{sec:numerical}
Let $X \sim N(0, 1)$ and $\epsilon = (\epsilon_Y(1), \epsilon_Y(0), \epsilon_R(1), \epsilon_R(0))' \sim N(0, \Sigma)$, where the diagonal elements of $\Sigma$ are $1$ and all off-diagonal elements are $-0.3$. Suppose for $d \in \{0, 1\}$,
\begin{align*}
    Y(d) & = \mu_d(X) + \epsilon_Y(d) \\
    R(d) & = I \{\epsilon_R(d) \leq \nu_d(X)\}~,
\end{align*}
with $\mu_d(x)$ and $\nu_d(x)$ specified below. In the following two examples, the values of $\theta$ can be calculated by hand, and the values of $\theta^{\rm obs}$ and $\theta^{\rm drop}$ are computed via simulation with $n = 10^6$ random draws.
\begin{enumerate}
    \item $\mu_1(x) = 2x$, $\mu_0(x) = x^3$, $\nu_1(x) = x$, $\nu_0(x) = x^2$. In this example, $\theta = 0$, $\theta^{\rm obs} \approx 1.17$, $\theta^{\rm drop} \approx -0.50$.
    \item $\mu_1(x) = 2x$, $\mu_0(x) = x$, $\nu_1(x) = x$, $\nu_0(x) = x$. In this example, $\theta = 0$, $\theta^{\rm obs} \approx 0.56$, $\theta^{\rm drop} \approx 0.86$.
\end{enumerate}

\subsection{Additional Details for Empirical Survey in Section \ref{sec:application-aer-aej}} \label{sec:addl-info-empirical}

\begin{table}[H]
  \centering
  \caption{Additional notes about each paper used in Figure \ref{fig:comparison}}
  \begin{adjustbox}{max width=\linewidth,center}
    \begin{tabular}{p{0.25\linewidth} p{0.3\linewidth} p{0.45\linewidth}} 
    \toprule
    Paper    & \multicolumn{1}{c}{Table Replicated} & \multicolumn{1}{c}{Additional Notes} \\
    \midrule
    \cite{dhar2022reshaping}     & Table 2: (1), (2) and (3) & Original specification features controls. Original estimates do not include strata fixed-effects. \\
    \midrule
    \cite{carter2021subsidies}   & Figure 2: left panel (``Direct impact on treatment group") &
      Original specification features controls. Original estimates include strata fixed-effects. We reported both ``During" and ``After" estimates. \\
    \midrule
    \cite{casaburi2021using}     & Table 2: (1)              & Original specification does not feature controls. Original estimate includes strata fixed-effects. \\
    \midrule
    \cite{abebe2021selection}    & Table 2, Table 3 (Column 1)                  & Original specification does not feature controls. Original estimates include strata fixed-effects.          \\
    \midrule
    \cite{hjort2021research}     &
      Online Appendix Table A.11: (1) &
      Original specification does not feature controls. Original estimate does not include strata fixed-effects. This is an intent-to-treat specification. \\
    \midrule
    \cite{romero2020outsourcing} &
      Table 3: (4) &
      Original specification does not feature controls. Original estimates include pair fixed-effects. These are intent-to-treat specifications. \\
    \midrule
    \cite{attanasio2020estimating} &
      Table 4: Second Column &
      Original specification features controls. Original estimate does not include strata fixed-effects. The first column of Table 4 is estimated using a probit regression and thus is not reproduced. \\
    \bottomrule
    \end{tabular}
    \end{adjustbox}
\begin{tablenotes} \footnotesize 
\item Notes: For each paper considered in Section \ref{sec:application-aer-aej}, we list the corresponding table/figure and specification(s) replicated in the second column. We include relevant notes for each application in the third column. 
\end{tablenotes}
\label{table:addl-details-empirical-applications}\end{table}%

\subsection{Details for Equation \eqref{eq:pfe}}\label{sec:FWL}
Let $\tilde \theta_n^{\rm drop}$ denote the OLS estimator of $\theta^{\rm drop}$ in \eqref{eq:pfe} using only observations with $R_i = 1$. By construction, the $j$th entry of the OLS estimator of the projection coefficient of $D_i$ on the pair fixed effects is given by
\begin{equation}\label{eq:projection}
\left ( \sum_{1 \leq i \leq n: R_i = 1} I \{i \in \{\pi(2j - 1), \pi(2j)\} \right )^{-1} \sum_{1 \leq i \leq n: R_i = 1} D_i I \{i \in \{\pi(2j - 1), \pi(2j)\}\}~. 
\end{equation}
Let $\tilde D_i$ denote the residual from the projection of $D_i$ on the pair fixed effects. Fix $1 \leq j \leq n$. If $R_{\pi(2j - 1)} = R_{\pi(2j)} = 1$, then it follows from \eqref{eq:projection} that 
\[ \tilde D_{\pi(2j)} = \frac{1}{2}\left(D_{\pi(2j)} - D_{\pi(2j - 1)}\right)~, \]
\[ \tilde D_{\pi(2j - 1)} = \frac{1}{2}\left(D_{\pi(2j - 1)} - D_{\pi(2j)}\right)~. \]
Next suppose the $j$th pair contains only one attrited unit. Without loss of generality, assume $R_{\pi(2j - 1)} = 0$ and $R_{\pi(2j)} = 1$. 
It then follows from \eqref{eq:projection} that 
\[ \tilde D_{\pi(2j)} = D_{\pi(2j)} - D_{\pi(2j)} = 0~. \]
By an application of the Frisch-Waugh-Lovell theorem we can thus conclude that $\tilde \theta_n^{\rm drop} = \hat \theta_n^{\rm drop}$, as desired.

\subsection{Relevant Excerpts from Referenced Sources}\label{sec:quotes}
\cite{donner2000design} chapter 3, page 40:

``A final disadvantage of the matched pair design is that the loss to follow-up of a single cluster in a pair implies that both clusters in that pair must effectively be discarded from the trial, at least with respect to testing the effect of intervention. This problem [...] clearly does not arise if there is some replication of clusters within each combination of intervention and stratum.''


\noindent \cite{king2007politically} page 490:

``The key additional advantage of the matched pair design from our perspective is that it enables us to protect ourselves, to a degree, from selection bias that could otherwise occur with the loss of clusters. In particular, if we lose a cluster for a reason related to one or more of the variables we matched on [...] then no bias would be induced for the remaining clusters. That is, whether we delete or impute the remaining member of the pair that suffered a loss of a cluster under these circumstances, the set of all remaining pairs in the study would still be as balanced—matched on observed background characteristics and randomized within pairs—as the original full data set. Thus, any variable we can measure and match on when creating pairs removes a potential for selection bias if later on we lose a cluster due to a reason related to that variable. [...] Classical randomization, which does not match on any variables, lacks this protective property.''

\noindent \cite{bruhn2009pursuit} page 209:

``King et al. (2007) emphasize one additional advantage in the context of social science experiments when the matched pairs occur at the level of a community, village, or school, which is that it provides partial protection against political interference or drop-out. If a unit drops out of the study [...] its pair unit can also be dropped from the study, while the set of remaining pairs will still be as balanced as the original dataset. In contrast, in a pure randomized experiment, if even one unit drops out, it is no longer guaranteed that the treatment and control groups are balanced, on average.''

\noindent \cite{glennerster2013running} chapter 4, page 159:

``In paired matching, for example, if we lose one of the units in the pair [...] and we include a dummy for the stratum, essentially we have to drop the other unit in the pair from the analysis. [...] Some evaluators have mistakenly seen this as an advantage of pairing [...] But in fact if we drop the pair we have just introduced even more attrition bias. [...] Our suggestion is that if there is a risk of attrition [...] use strata that have at least four units rather than pairwise randomization.''

\end{document}